\documentclass[11pt,a4paper]{article}
\usepackage[english]{babel}
\usepackage{german}
\usepackage{amssymb}
\usepackage{amsfonts}
\usepackage{amsmath}
\usepackage[latin1]{inputenc}
\usepackage{fullpage}
\usepackage{graphicx}
\usepackage{subfigure}
\usepackage{algorithm}
\usepackage[noend]{algorithmic}
\usepackage{amsmath,amssymb,cite}
\usepackage{colortbl}
\usepackage{vmargin}

\renewcommand{\algorithmicrequire}{\textbf{Input:}}
\renewcommand{\algorithmicensure}{\textbf{Output:}}
\newtheorem{theorem}{Theorem}

\newtheorem{corollary}{Corollary}

\newtheorem{definition}{Definition}

\newtheorem{observation}{Observation}

\newtheorem{lemma}{Lemma}

\newtheorem{proposition}{Proposition}

\newenvironment{proof}[1][Proof]{\noindent\textbf{#1.} }{\ \rule{0.5em}{0.5em}}

\newcommand{\bbbn}{\mathbb{N}}

\newcommand{\cc}{\mathcal{C}}

\newcommand{\cj}{\mathcal{J}}

\newcommand{\bs}{\backslash}

\newcommand{\ap}{\text{\emph{all paths}}}

\newcommand{\reach}{reach}
\newcommand{\ALG}{\text{ALG}}
\newcommand{\OPT}{\text{OPT}}

  \renewenvironment{thebibliography}[1]{
\begin{oldthebibliography}{#1}
\setlength{\parskip}{0.1ex} \setlength{\itemsep}{0.9ex}} {\end{oldthebibliography}}

    \setmarginsrb{3.5cm}{2.25cm}{3.5cm}{3.05cm}{0.3cm}{0.3cm}{-0.3cm}{1.0cm}

\begin{document}

\title{\Large Temporal Network Optimization Subject to Connectivity Constraints\footnote{This work was supported in part 
by (i) the project ``Foundations of Dynamic Distributed Computing Systems'' (\textsf{FOCUS}) which is implemented under 
the ``ARISTEIA'' Action of the  Operational Programme ``Education and Lifelong Learning'' and is co-funded by the 
European Union (European Social Fund) and Greek National Resources, (ii) the FET EU IP project \textsf{MULTIPLEX} under 
contract no 317532, and (iii) the EPSRC Grants EP/P020372/1, EP/P02002X/1, and EP/K022660/1. A preliminary version of this work has appeared at ICALP 2013 \cite{MMCS13}.}}

\author{George B. Mertzios\thanks{Department of Computer Science, Durham University, UK. 
Email: \texttt{george.mertzios@durham.ac.uk}}
\and 
Othon Michail\thanks{Department of Computer Science, University of Liverpool, UK. 
Email: \texttt{Othon.Michail@liverpool.ac.uk}}
\and 
Paul G. Spirakis\thanks{Department of Computer Science, University of Liverpool, UK and Department of Computer Engineering and Informatics, University of Patras, Greece. 
Email: \texttt{P.Spirakis@liverpool.ac.uk}}}
\date{\vspace{-0.8cm}}

\maketitle

\begin{abstract}
In this work we consider \emph{temporal networks}, i.e.~networks defined by a \emph{labeling} $\lambda$ assigning to each edge of an \emph{underlying graph} $G$ a set of \emph{discrete} time-labels. The labels of an edge, which are natural numbers, indicate the discrete time moments at which the edge is available. We focus on \emph{path problems} of temporal networks. In particular, we consider \emph{time-respecting} paths, i.e.~paths whose edges are assigned by $\lambda$ a strictly increasing sequence of labels. We begin by giving two efficient algorithms for computing shortest time-respecting paths on a temporal network. We then prove that there is a \emph{natural analogue of Menger's theorem} holding for arbitrary temporal networks. Finally, we propose two \emph{cost minimization parameters} for temporal network design. One is the \emph{temporality} of $G$, in which the goal is to minimize the maximum number of labels of an edge, and the other is the \emph{temporal cost} of $G$, in which the goal is to minimize the total number of labels used. Optimization of these parameters is performed subject to some \emph{connectivity constraint}. We prove several lower and upper bounds for the temporality and the temporal cost of some very basic graph families such as rings, directed acyclic graphs, and trees.\newline

\end{abstract}

\section{Introduction}
\label{sec:intro}

A \emph{temporal} (or \emph{dynamic}) \emph{network} is, loosely speaking, a network that changes with time. This notion encloses a great variety of both modern and traditional networks such as information and communication networks, social networks, transportation networks, and several physical systems. In the literature of traditional communication networks, the network topology is rather static, i.e.~topology modifications are rare and they are mainly due to link failures and congestion. 
However, most modern communication networks such as mobile ad hoc, sensor, peer-to-peer, opportunistic, and delay-tolerant networks are inherently dynamic and it is often the case that this dynamicity is of a very high rate. In social networks, the topology usually represents the social connections between a group of individuals and it changes as the social relationships between the individuals are updated, or as existing individuals leave, or new individuals enter the group. In a transportation network, there is usually some fixed network of routes and a set of transportation units moving over these routes and dynamicity refers to the change of the positions of the transportation units in the network as time passes. Physical systems of interest may include several systems of interacting particles. 

In this work, embarking from the foundational work of Kempe \emph{et al.}~\cite{KKK00}, we consider \emph{discrete time}, that is, we consider networks in which changes occur at discrete moments in time, e.g.~days. This choice is not only a very natural abstraction of many real systems but also gives to the resulting models a purely combinatorial flavor. In particular, we consider those networks that can be described via an underlying graph $G$ and a labeling $\lambda$ assigning to each edge of $G$ a (possibly empty) set of discrete labels. Note that this is a generalization of the single-label-per-edge model used in~\cite{KKK00}, as we allow many time-labels to appear on an edge. These labels are drawn from the natural numbers and indicate the discrete moments in time at which the corresponding connection is available. For example, in the case of a communication network, availability of a communication link at some time $t$ may mean that a communication protocol is allowed to transmit a data packet over that link at time $t$.

In this work, we initiate the study of the following fundamental network design problem: ``\emph{Given an underlying (di)graph $G$, assign labels to the edges of $G$ so that the resulting temporal graph $\lambda (G)$ minimizes some parameter while satisfying some connectivity property}''. 
In particular, we consider two cost optimization parameters for a given graph $G$. The first one, called \emph{temporality} of $G$, measures the maximum number of labels that an edge of $G$ has been assigned. The second one, called \emph{temporal cost} of $G$, measures the total number of labels that have been assigned to all edges of $G$  (i.e.~if $|\lambda(e)|$ denotes the number of labels assigned to edge $e$, we are interested in $\sum_{e\in E} |\lambda(e)|$). That is, if we interpret the number of assigned labels as a measure of \emph{cost}, the temporality (resp.~the temporal cost)\ of $G$ is a measure of the decentralized (resp.~centralized) cost of the network, where only the cost of individual edges (resp.~the total cost over all edges) is considered. Each of these two cost measures can be minimized subject to some particular connectivity property $\mathcal{P}$ that the temporal graph $\lambda (G)$ has to satisfy. In this work, we consider two very basic connectivity properties. The first one, that we call the \emph{all paths} property, requires the temporal graph to preserve every simple path of its underlying graph, where by ``preserve a path of $G$'' we mean in this work that the labeling should provide at least one strictly increasing sequence of labels on the edges of that path, in which case we also say that the path is \emph{time-respecting}. 

Before describing our second connectivity property let us give a simple illustration of temporality minimization. We are given a directed ring $u_1,u_2,\ldots,u_n$ and we want to determine the temporality of the ring subject to the all paths property. That is, we want to find a labeling $\lambda$ that preserves every simple path of the ring and at the same time minimizes the maximum number of labels of an edge. Looking at Figure~\ref{fig:ring}, it is immediate to observe that an increasing sequence of labels on the edges of path $P_1$ implies a decreasing pair of labels on edges $(u_{n-1},u_n)$ and $(u_1,u_2)$. On the other hand, path $P_2$ uses first $(u_{n-1},u_n)$ and then $(u_1,u_2)$ thus it requires an increasing pair of labels on these edges. It follows that in order to preserve both $P_1$ and $P_2$ we have to use a second label on at least one of these two edges, thus the temporality is at least 2. Next, consider the labeling that assigns to each edge $(u_i,u_{i+1})$ the labels $\{i, n+i\}$, where $1\leq i\leq n$ and $u_{n+1}=u_1$. It is not hard to see that this labeling preserves all simple paths of the ring. Since the maximum number of labels that it assigns to an edge is 2, we conclude that the temporality is also at most 2. In summary, the temporality of preserving all simple paths of a directed ring is 2.   

\begin{figure}[!hbtp]
\centering{
\includegraphics[width=0.4\textwidth]{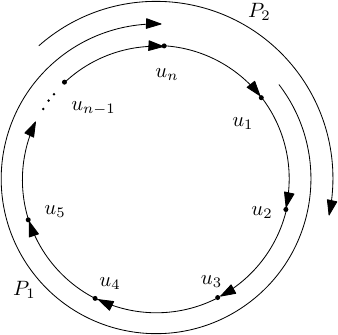}
}
\caption{Path $P_2$ forces a second label to appear on either $(u_{n-1},u_n)$ or $(u_1,u_2)$.} \label{fig:ring}
\end{figure}

The other connectivity property that we define, called the \emph{reach} property, requires the temporal graph to preserve a path from node $u$ to node $v$ whenever $v$ is reachable from $u$ in the underlying graph. Furthermore, the minimization of each of our two cost measures can be affected by some problem-specific constraints on the labels that we are allowed to use. We consider here one of the most natural constraints, namely an upper bound of the \emph{age} of the constructed labeling $\lambda$, where the age of a labeling $\lambda$ is defined to be equal to the maximum label of $\lambda$ minus its minimum label plus 1. Now the goal is to minimize the cost parameter, e.g.~the temporality, satisfy the connectivity property, e.g.~\emph{all paths}, and additionally guarantee that the age does not exceed some given natural $k$. Returning to the ring example, it is not hard to see, that if we additionally restrict the age to be at most $n-1$ then we can no longer preserve all paths of a ring using at most 2 labels per edge. In fact, we must now necessarily use the worst possible number of labels, i.e.~$n-1$ on every edge. 

Minimizing such parameters may be crucial as, in most real networks, making a connection available and maintaining its availability does not come for free. For example, in wireless sensor networks the cost of making edges available is directly related to the power consumption of keeping nodes awake, of broadcasting, of listening to the wireless channel, and of resolving the resulting communication collisions. The same holds for transportation networks where the goal is to achieve good connectivity properties with as few transportation units as possible. At the same time, such a study is important from a purely graph-theoretic perspective as it gives some first insight into the structure of specific families of temporal graphs. To make this clear, consider again the ring example. Proving that the temporality of preserving all paths of a ring is 2 at the same time proves the following. If a \emph{temporal ring} is defined as a ring in which all nodes can communicate clockwise to all other nodes via time-respecting paths then \emph{no temporal ring exists with fewer than $n+1$ labels}. This, though an easy one, is a structural result for temporal graphs. Finally, we believe that our results are a first step towards answering the following fundamental question: ``\emph{To what extent can algorithmic and structural results of graph theory be carried over to temporal graphs?}''. For example, is there an analogue of Menger's theorem for temporal graphs? One of the results of the present work is an affirmative answer to the latter question.

\subsection{Related Work}
\label{subsec:rw}

\noindent\textbf{Labeled Graphs.} Labeled graphs have been widely used in Computer Science and Mathematics, e.g.~in Graph Coloring \cite{MR02}. In our work, labels correspond to moments in time and the properties of labeled graphs that we consider are naturally \emph{temporal properties}. Note, however, that any property of a graph labeled from a discrete set of labels corresponds to some temporal property if interpreted appropriately. For example, a proper edge-coloring, i.e.~a coloring of the edges in which no two adjacent edges share a common color, corresponds to a temporal graph in which no two adjacent edges share a common label, i.e.~no two adjacent edges ever appear at the same time. Though we focus on properties with natural temporal meaning, our definitions are generic and do not exclude other, yet to be defined, properties that may prove important in future applications.  

\noindent\textbf{Single-label Temporal Graphs and Menger's Theorem.} The model of temporal graphs that we consider in this work is a direct extension of the single-label 
model studied in \cite{Be96} and \cite{KKK00} to allow for many labels per edge. The main result of \cite{Be96} was that in single-label 
networks the max-flow min-cut theorem holds with unit capacities for time-respecting paths. In \cite{KKK00}, Kempe \emph{et al.}, among 
other things, proved that a fundamental property of classical graphs does not carry over to 
their temporal counterparts. In particular, they proved that there is no analogue of Menger's theorem, at least in its original formulation, 
for arbitrary single-label temporal networks and that the computation of the number of node-disjoint $s$-$t$ time-respecting paths is NP-complete. 
\emph{Menger's theorem} 
states that the maximum number of node-disjoint $s$-$t$ paths is equal to the minimum number of nodes 
needed to separate $s$ from $t$ (see \cite{Bo98}). In this work, we go a step ahead showing that if one reformulates Menger's theorem 
in a way that takes time into account then a very natural temporal analogue of Menger's theorem is obtained. Both of the above papers, 
consider a path as \emph{time-respecting} if its edges have non-decreasing labels. In the present work, we depart from this assumption and 
consider a path as time-respecting if its edges have \emph{strictly increasing} labels. Our choice is very well motivated by recent work in 
dynamic communication networks. If it takes one time unit to transmit a data packet over 
a link then a packet can only be transmitted over paths with strictly increasing availability times. 

\noindent\textbf{Continuous Availabilities (Intervals).} Some authors have assumed that an edge may be available for a whole time-interval $[t_1,t_2]$ or several such intervals and not just for discrete moments as we assume here. This is a clearly natural assumption but the techniques used in those works are quite different from those needed in the discrete case \cite{XFJ03,FT98}. 

\noindent\textbf{Dynamic Distributed Networks.} In recent years, there is a growing interest in distributed computing systems that are inherently dynamic. This has been mainly driven by the advent of low-cost wireless communication devices and the development of efficient wireless communication protocols. Apart from the huge amount of work that has been devoted to applications, there is also a steadily growing concrete set of foundational work. A notable set of works has studied (distributed) computation in \emph{worst-case} dynamic networks in which the topology may change arbitrarily from round to round subject to some constraints that allow for bounded end-to-end communication~\cite{OW05,KLO10,MCS12b-journal,DPRS13}. Population protocols \cite{AADFP06} and variants \cite{MCS11-2} are collections of finite-state agents that move arbitrarily like a soup of particles and interact in pairs when they come close to each other. The goal is there for the population to compute (i.e.~agree on) something useful in the limit in such an adversarial setting. Another interesting direction assumes that the dynamicity of the network is a result of randomness. Here the interest is on determining ``good'' properties of the dynamic network that hold with high probability, such as small (temporal) diameter, and on designing protocols for distributed tasks \cite{CFTE08,AKL08}. For introductory texts on the above lines of research in dynamic distributed networks the reader is referred to \cite{CFQS12,MCS11,Sc02}.

\noindent\textbf{Distance Labeling.} A distance labeling of a graph $G$ is an assignment of unique labels to the vertices of $G$ so that the distance between any two vertices can be inferred from their labels alone. The goal is to minimize some parameter of the labeling and to provide a (hopefully fast) decoder algorithm for extracting a distance from two labels \cite{GPPR01,KKKP04}. There are several differences between a distance labeling and the time-labelings that we consider in this work. First of all, a distance labeling is being assigned on the vertices and not on the edges. Moreover, in distance labeling, one usually seeks the most compact set of labels (in binary length) that still guarantees efficient decoding. That is, the labeling parameter to be minimized is the binary length of an appropriate encoding, which is quite different from our cost parameters. Finally, the optimization constraint there is efficient decoding while in our case the constraints have to do with connectivity properties of the labeled graph. 

Also, we encourage the interested reader to see \cite{Mi16} for a recent introductory text on the recent algorithmic progress on temporal graphs.

\subsection{Contribution}
\label{subsec:con}

In \S \ref{sec:prel}, we formally define the model of temporal graphs under consideration and provide all further necessary definitions. The rest of the paper is partitioned into two parts. Part I focuses on journey problems for temporal graphs. In particular, in \S \ref{sec:journeys}, we give two efficient algorithms for computing shortest time-respecting paths. Then in \S \ref{sec:menger} we present an analogue of Menger's theorem which we prove valid for arbitrary temporal graphs. We apply our Menger's analogue to simplify the proof of a recent result on distributed token gathering. Part II studies the problem of designing a temporal graph optimizing some parameters while satisfying some connectivity constraints. Specifically, in \S \ref{min-cost-connectivity-sec} we formally define the temporality and temporal cost optimization metrics for temporal graphs. In \S \ref{basic-properties-cost-subsec}, we provide several upper and lower bounds for the temporality of some fundamental graph families such as rings, directed acyclic graphs (DAGs), and trees, as well as an interesting trade-off between the temporality and the age of rings. Furthermore, we provide in \S \ref{generic-method-subsec} a generic method for computing a lower bound of the temporality of an arbitrary graph w.r.t. the $\ap$ property, and we illustrate its usefulness in cliques, close-to-complete bipartite subgraphs, and planar graphs. In \S \ref{cost-computation-subsec}, we consider the temporal cost of a digraph $G$ w.r.t. the $\reach$ property, when additionally the age of the resulting labeling $\lambda (G)$ is restricted to be the smallest possible. We prove that this problem is hard to approximate, i.e.~there exists no PTAS unless P=NP. To prove our claim, we first prove (which may be of interest in its own right) that the Max-XOR($3$) problem is APX-hard via a PTAS reduction from Max-XOR. In the Max-XOR($3$) problem, we are given a $2$-CNF formula $\phi $, every literal of which appears in at most 3 clauses, and we want to compute the greatest number of clauses of $\phi $ that can be simultaneously XOR-satisfied. Then we provide a PTAS reduction from Max-XOR$(3)$ to our temporal cost minimization problem. On the positive side, we provide an $(r(G)/n)$-factor approximation algorithm for the latter problem, where $r(G)$ denotes the total number of reachabilities in $G$. Finally, in \S \ref{sec:conc} we conclude and give further research directions that are opened by our work.

\section{Preliminaries}
\label{sec:prel}

\subsection{A Model of Temporal Graphs}

Given a (di)graph $G=(V,E)$, \footnote{The reason that we do not consider only digraphs and then allow undirected graphs to result as their special case, is that in that way an undirected edge would formally consist of two antiparallel edges. This would allow those edges to be labeled differently, unless we introduced an additional constraint preventing it. We've chosen to avoid this by considering explicit undirected graphs (whenever required) with at most one bidirectional edge per pair of nodes.} a \emph{labeling} of $G$ is a mapping $\lambda:E\rightarrow 2^\bbbn$, that is, a labeling assigns to each edge of $G$ a (possibly empty) \footnote{The reader may be wondering whether it is pointless to allow the assignment of no labels to an edge $e$ of $G$, as it would have been equivalent to delete $e$ from $G$ in the first place. Even though this is true for temporal graphs provided as input, it isn't for temporal graphs that will be \emph{designed} by an algorithm based on an underlying graph. In the latter case, it is the algorithm's task to decide whether some of the provided edges need not be ever made available.} set of natural numbers, called \emph{labels}.

\begin{definition}
\label{temporal-graph-def}
Let $G=(V,E)$ be a (di)graph and $\lambda$ be a labeling of $G$. Then $\lambda (G)$ is the \emph{temporal graph} (or \emph{dynamic graph} \footnote{Even though both names are almost equally used in the literature, in this paper we have chosen to use the term ``temporal'' in order to avoid confusion of readers that are more familiar with the use of the term ``dynamic'' to refer to dynamically updated instances, with which usually an algorithm has to deal in an online way (including the rich literature of problems in which the algorithm has to maintain a graph property that is being disturbed by adversarial graph modifications).}) of $G$ with respect to $\lambda$. Furthermore, $G$ is the \emph{underlying graph} of $\lambda (G)$.
\end{definition}

We denote by $\lambda(E)$ the multiset of all labels assigned to the underlying graph by the labeling $\lambda$ and by $|\lambda|= |\lambda(E)|$ their cardinality (i.e.~$|\lambda|=\sum_{e\in E} |\lambda(e)|$). We also denote by $\lambda_{\min}= \min\{l\in \lambda(E)\}$ the minimum label and by $\lambda_{\max}= \max\{l\in \lambda(E)\}$ the maximum label assigned by $\lambda$. We define the \emph{age} of a temporal graph $\lambda(G)$ as $\alpha(\lambda)= \lambda_{\max}-\lambda_{\min}+1$. Note that in case $\lambda_{\min}=1$ then we have $\alpha(\lambda)=\lambda_{\max}$. For every graph $G$ we denote by $\mathcal{L}_{G}$ the set of all possible labelings $\lambda $ of $G$. Furthermore, for every $k\in \mathbb{N}$, 
we define $\mathcal{L}_{G,k}=\{\lambda \in \mathcal{L}_{G}:\alpha (\lambda )\leq k\}$.

\subsection{Further Definitions}

For every time $r\in\bbbn$, we define the $r$\emph{th instance of a temporal graph $\lambda(G)$} as the static graph $\lambda(G,r)= (V,E(r))$, where $E(r)= \{e\in E: r\in\lambda(e)\}$ is the (possibly empty) set of all edges of the underlying graph $G$ that are assigned label $r$ by labeling $\lambda$. A temporal graph $\lambda(G)$ may be also viewed as a \emph{sequence of static graphs} $(G_1,G_2,\ldots,G_{\alpha(\lambda)})$, where $G_i=\lambda(G,\lambda_{\min}+i-1)$ for all $1\leq i\leq\alpha(\lambda)$. 
Another, often convenient, representation of a temporal graph is the following.

\begin{definition}
The \emph{static expansion} \footnote{The notion of static expansion is related to the notion of \emph{time-expanded graphs} of temporal graphs such as periodic, or resulting from public transportation networks (cf. \cite{SWW00,MSWZ07}).} of a temporal graph $\lambda(G)$ is a \emph{static digraph} $H=(S,A)$, and in particular a DAG, defined as follows. If $V=\{u_1,u_2,\ldots,u_n\}$ then $S= \{u_{ij}: \lambda_{\min}-1\leq i\leq \lambda_{\max},1\leq j\leq n\}$ and $A=\{(u_{(i-1)j},u_{ij^\prime}):$ if $j=j^\prime$ or $(u_j,u_j^\prime)\in E(i)$ for some $\lambda_{\min}\leq i\leq \lambda_{\max}\}$. In words, we create $\alpha(\lambda)+1$ copies of $V$ representing the nodes over time (time-nodes) and add outgoing edges from time-nodes of one level only to time-nodes of the next level. In particular, we connect a time-node $u_{(i-1)j}$ to its own subsequent copy $u_{ij}$ and to every time node $u_{ij^\prime}$ s.t.~$(u_j,u_j^\prime)$ is an edge of $\lambda(G)$ at time $i$.  
\end{definition}

A \emph{journey} (or \emph{time-respecting path}) $J$ of a temporal graph $\lambda(G)$ is a path $(e_1,e_2,$ $\ldots,e_k)$ of the underlying graph $G=(V,E)$, where $e_i\in E$, together with labels $l_1<l_2<\ldots<l_k$ such that $l_i\in \lambda(e_i)$ for all $1\leq i\leq k$. 
In words, a journey is a path that uses strictly increasing edge-labels. If labeling $\lambda$ defines a journey on some path $P$ of $G$ then we also say that $\lambda $ \emph{preserves} $P$. 
A natural notation for a journey is $(e_1,l_1),(e_2,l_2),\ldots,(e_k,l_k)$. We call each $(e_i,l_i)$ a \emph{time-edge} as it corresponds to the availability of edge $e_i$ at some time $l_i$. We call $l_1$ the \emph{departure time} and $l_k$ the \emph{arrival time} of journey $J$ and denote them by $d(J)$ and $a(J)$, respectively. A $(u,v)$-journey $J$ is called \emph{foremost from time $t$} if $d(J)\geq t$ and $a(J)$ is minimized. Formally, let $\cj$ be the set of all $(u,v)$-journeys $J$ with $d(J)\geq t$. A $J\in \cj$ is foremost if $a(J)=\min_{J^\prime\in \cj}\{ a(J^\prime)\}$. A journey $J$ is called \emph{fastest} if $a(J)-d(J)+1$ is minimized. We call $a(J)-d(J)+1$ the \emph{duration} of the journey. A journey $J$ is called \emph{shortest} if $k$ is minimized, that is it minimizes the number of nodes visited (also called number of hops).

We say that a journey $J$ \emph{leaves from node $u$} (\emph{arrives at node $u$}, resp.) \emph{at time $t$} if $(u,v,t)$ ($(v,u,t)$, resp.) is a time-edge of $J$. Two journeys are called \emph{out-disjoint} (\emph{in-disjoint}, respectively) if they never leave from (arrive at, resp.) the same node at the same time.

Given a set $\cj$ of $(s,v)$-journeys we define their \emph{arrival time} as $a(\cj)=\max_{J\in \cj}$ $\{a(J)\}$. We say that a set $\cj$ of $(s,v)$-journeys satisfying some constraint $c$ (e.g.~containing at least $k$ journeys and/or containing only out-disjoint journeys) is \emph{foremost} if $a(\cj)$ is minimized over all sets of journeys satisfying the constraint.

If, in addition to the labeling $\lambda$, a positive weight $w(e)>0$ is assigned to every edge $e\in E$, then we call a temporal graph a \emph{weighted} temporal graph. In case of a weighted temporal graph, by ``shortest journey'' we mean a journey that minimizes the sum of the weights of its edges.

Throughout the text we denote by $n$ the number of nodes and by $m$ and $m_t$ the number of edges of graphs and temporal graphs, respectively. In case of a temporal graph, by ``number of edges'' we mean ``number of time-edges'', i.e.~$m_t=|\lambda|$. By $d(G)$ we denote the diameter of a (di)graph $G$, that is the length of the longest shortest path between any two nodes of $G$. By $\delta_u$ we denote the degree of a node $u\in V(G)$ (in case of an undirected graph $G$). 

\part*{Part I}

\section{Journey Problems}
\label{sec:journeys}

\subsection{Foremost Journeys}

We are given (in its full ``offline'' description) a \emph{temporal graph} $\lambda(G)$, where $G=(V,E)$, a distinguished source node $s\in V$, and a time $\lambda_{\min}\leq t_{start}\leq \lambda_{\max}$ and we are asked for all $w\in V\bs\{s\}$ to compute a foremost $(s,w)$-journey from time $t_{start}$.

\begin{algorithm}[t!]
\caption{FJ} \label{alg:fj}
\begin{algorithmic} [1]
    \REQUIRE Temporal graph $\lambda(G)$ (full ``offline'' description), source node $s\in V$, and time $t_{start}$, where $\lambda_{\min}\leq t_{start}\leq \lambda_{\max}$. 
    The input is represented by an array $A_{v}$ with $\lambda_{\max}-\lambda_{\min}+1$ entries for every node $v$, 
    where the entry $A_{v}[t]$ stores a pointer to the linked list of the adjacent nodes of~$v$ at time step $t$.
    \ENSURE For all $v\in V\bs\{s\}$ a foremost $(s,v)$-journey from time $t_{start}$. In particular, outputs for every $v$ a pair $(p[v],a[v])$, 
where $p[v]$ is the predecessor node of $v$ on the journey and $a[v]$ is the arrival time of the journey at $v$ (the pair as a whole may be viewed as the predecessor time-node of $v$ on the journey).

    \medskip
    
     \STATE{$R\leftarrow\{s\}$, $t\leftarrow t_{start}$}
     \FOR{each $v\in V\bs\{s\}$}
          \STATE $p[v]\leftarrow\emptyset$
          \STATE $a[v]\leftarrow\infty$
     \ENDFOR
     
     \medskip
     
     \WHILE{$R\neq V$ and $t\neq \lambda_{\max}+1$}
	      \STATE{$C\leftarrow \emptyset$}
          \FOR{each $u\in R$}
               \FOR{each $(u,v)\in E(t)$}
                    \IF[that is, $v\notin R$]{$p[v]=\emptyset$}
                         \STATE{$p[v]\leftarrow u$}
                         \STATE{$a[v]\leftarrow t$}
                         \STATE{$C\leftarrow C\cup\{v\}$}
                    \ENDIF
               \ENDFOR
          \ENDFOR
          \STATE{$R\leftarrow R \cup C$}
          \STATE{$t++$}
    \ENDWHILE
\end{algorithmic}
\end{algorithm}

\begin{theorem}
Algorithm \ref{alg:fj} correctly computes for all $w\in V\bs\{s\}$ a foremost $(s,w)$-journey from time $t_{start}$. 
The running time of the algorithm is $O(n\lambda_{\max}+m_t)$.
\end{theorem}
\begin{proof}
Assume that at the end of round $t-1$ all nodes in $R$ have been reached by foremost journeys from $s$. Let $(u,v,t)$ be a time-edge s.t.~$u\in R$ and $v\notin R$ and let $f(s,u)$ denote the foremost journey from $s$ to $u$. We claim that $J=f(s,u),(u,v,t)$ is a foremost journey from $s$ to $v$. Recall that we denote the arrival time of $J$ by $a(J)$. To see that our claim holds assume that there is some other journey $J^\prime$ s.t.~$a(J^\prime)<a(J)$. So there must be some time-edge $(w,z,t^\prime)$ for $w\in R$, $z\notin R$ and $t^\prime<t$. However, this contradicts the fact that $z\notin R$ as the algorithm should have added it in $R$ at time $t^\prime$. The proof follows by induction on $t$ beginning from $t=t_{start}$ at which time $R=\{s\}$ ($s$~has trivially been reached by a foremost journey from itself so the claim holds for the base case).

We now prove that the time complexity of the algorithm is $O(n\lambda_{\max}+m_t)$. 
In the worst-case, the last node may be inserted at step $\lambda_{\max}$, so the while loop is executed $O(\lambda_{\max})$ times. 
In each execution of the while loop, the algorithm visits the $O(n)$ nodes of the current set $R$ in the worst-case 
(e.g.~when all nodes but one have been added into $R$ from the first step). 
For each such node $v$ and for each time $\lambda_{\min} \leq t\leq \lambda_{\max}$ the algorithm first 
locates the entry $A_{v}[t]$ in the array $A_{v}$ in constant time and then it visits the whole linked list of the adjacent nodes of $v$ at time step $t$. 
All these operations can be performed in $O(n\lambda_{\max}+m_t)$ time in total.
\qquad
\end{proof}

\subsection{Shortest Journeys with Weights}

\begin{theorem}
\label{shortest-time-path-thm}Let $\lambda(G)$, where $G=(V,E)$, be a weighted temporal graph
with $n$ vertices and $m$ edges. Assume also that $|\lambda(e)|=1$ for all $e\in E$, i.e.~there is a single label on each edge (this implies also that $m_t=m$). Let $s,t\in V$. Then, we can compute a shortest journey $J$ between $s$ and~$t$ in $\lambda(G)$ (or report that no such journey exists) in $O(m\log
m+\sum_{v\in V}\delta_{v}^{2})=O(n^{3})$ time, where $\delta_{v}$ is the degree of~$v$
in $\lambda(G)$.
\end{theorem}

\begin{proof}
First, we may assume without loss of generality that $\lambda(G)$ is a connected
graph, and thus $m\geq n-1$. For the purposes of the proof we construct from 
$\lambda(G)$ a weighted directed graph $H$ with two specific vertices $s^{\prime
},t^{\prime }$, such that there exists a journey $J$ in $\lambda(G)$
between $s$ and $t$ if and only if there is a directed path $P$ in 
$H$ from $s^{\prime }$ to $t^{\prime }$. Furthermore, if such paths exist,
then the weight of the shortest journey $J$ of $\lambda(G)$ between $s$
and $t$ equals the weight of the shortest directed path $P$ of $%
H$ from $s^{\prime }$ to $t^{\prime }$.

First consider the (undirected) graph $G^{\prime }$ that we obtain when we
add two vertices $s_{0}$ and $t_{0}$ to $\lambda(G)$ and the edges $s_{0}s$ and $tt_{0}$. Assign to these two new edges the weight zero and assign to them the
time labels $\lambda (s_{0}s)=0$ and $\lambda (tt_{0})=\lambda_{\max}+1$. 
Then, clearly there exists a time-respecting path between $s$ and $t
$ in $\lambda(G)$ if and only if there exists a time-respecting path between $s_{0}$
and $t_{0}$ in $G^{\prime }$, while the weights of these two paths coincide.
For simplicity of the presentation, denote in the following by $V$ and $E$
the vertex and edge sets of $G^{\prime }$, respectively. Then we construct $%
H=(V_{H},E_{H})$ from $G^{\prime }=(V,E)$ as follows. Let $V_{H}=E$.
Furthermore, for every vertex $v\in V$, denote by $M(v)=\{vu:u\in N(v)\}$
the set of all incident edges to $v$ in $G^{\prime }$. For every pair $%
e_{1},e_{2}\in M(v)$ for some $v\in V$, add the arc $\widehat{e_{1}e_{2}}$
to $E_{H}$ if and only if $\lambda (e_{1})< \lambda (e_{2})$. In this case, we
assign to the arc $\widehat{e_{1}e_{2}}$ of $E_{H}$ the weight $w_{H}(%
\widehat{e_{1}e_{2}})=w(e_{2})$.

Suppose first that $G^{\prime }$ has a journey between $s_{0}$
and $t_{0}$. Let $J=(u_{0},u_{1},\ldots ,u_{k})$, where $u_{0}=s_{0}$ and $%
u_{k}=t_{0}$, be the shortest among them with respect to the weight function 
$w$ of $G^{\prime }$. Then, by the definition of $G^{\prime }$, $s_{0}s$ and 
$tt_{0}$ are the first and the last edges of $J$. Furthermore, by the
definition of a time-respecting path, $\lambda (u_{i-1}u_{i})< \lambda
(u_{i}u_{i+1})$ for every $i=1,2,\ldots ,k-1$. Therefore, by the above
construction of $H$, there exists the directed path $Q=(e_{0},e_{1},\ldots
,e_{k-1})$ in $H$, where $e_{i}=u_{i}u_{i+1}$ for every $i=0,1,\ldots ,k-1$.
Note that $e_{0}=s_{0}s$ and that $e_{k-1}=tt_{0}$. Furthermore, in the
weight function $w_{H}$ of $H$, $w_{H}(\widehat{e_{i}e_{i+1}})=w(e_{i+1})$
for every $i=0,1,\ldots ,k-2$. Note that $w_{H}(\widehat{e_{k-2}e_{k-1}}%
)=w(e_{k-1})=w(u_{k-1}u_{k})$, i.e.~$w_{H}(\widehat{e_{k-2}e_{k-1}}%
)=w(tt_{0})=0$. Thus, the total weight $w(J)$ of $J$ in $G^{\prime }$ equals
the total weight $w_{H}(Q)$ of $Q$ in $H$.

Let now $s_{H}=s_{0}s$ and $t_{H}=tt_{0}$. Suppose now that $H$ has a path 
between $s_{H}$ and $t_{H}$. Let $Q=(e_{0},e_{1},\ldots
,e_{k})$, where $e_{0}=s_{H}$ and $e_{k}=t_{H}$, be the shortest among them
with respect to the weight function $w_{H}$ of $H$. Since $Q$ is a directed
path between $s_{H}$ and $t_{H}$, $\lambda (e_{i})< \lambda (e_{i+1})$ for
every $i=0,1,\ldots ,k-1$ by the construction of $H$. Furthermore, the edges 
$e_{i}$ and $e_{i+1}$ of $G^{\prime }$ are incident for every $i=0,1,\ldots
,k-1$. Denote now by $p_{i}$ the common vertex of the edges $e_{i}$ and $%
e_{i+1}$ in $G^{\prime }$ for every $i=0,1,\ldots ,k-1$. We will prove that $%
p_{i}\neq p_{i+1}$ for every $i=0,1,\ldots ,k-2$. Suppose otherwise that $%
p_{i}=p_{i+1}$ for some $0\leq i\leq k-2$. Then the edges $e_{i}$, $e_{i+1}$, 
and $e_{i+2}$ of $G^{\prime }$ are as it is shown in Figure~\ref{shortest-respecting-forbidden-fig}, where $e_{i}=ad$, $e_{i+1}=bd$, 
$e_{i+2}=cd$, and $d=p_{i}=p_{i+1}$ is the common point of the edges $e_{i}$, 
$e_{i+1}$, and $e_{i+2}$. However, since $\lambda (e_{i})< \lambda (e_{i+1})$
and $\lambda (e_{i+1})< \lambda (e_{i+2})$, it follows that $\lambda (e_{i})<
\lambda (e_{i+2})$, and thus there exists the arc $\widehat{e_{i}e_{i+2}}$ in
the directed graph $H$. Furthermore $w_{H}(\widehat{e_{i}e_{i+2}})=w_{H}(%
\widehat{e_{i+1}e_{i+2}})=w(e_{i+2})$, and thus $w_{H}(\widehat{e_{i}e_{i+1}}%
)+w_{H}(\widehat{e_{i+1}e_{i+2}})>w_{H}(\widehat{e_{i}e_{i+2}})$. Therefore
there exists in $H$ the strictly shorter directed path $Q^{\prime
}=(e_{0},e_{1},\ldots ,e_{i},e_{i+2}\ldots ,e_{k})$ between $e_{0}=s_{H}$
and $e_{k}=t_{H}$. This is a contradiction, since $Q$ is the shortest
directed path between $s_{H}$ and $t_{H}$. Therefore $p_{i}\neq p_{i+1}$ for
every $i=0,1,\ldots ,k-2$. Thus, we can denote now $e_{i}=p_{i-1}p_{i}$ for
every $i=1,2,\ldots ,k$, where $p_{0}=s_{0}$ and $p_{k}=t_{0}$. That is, $%
J=(p_{0},p_{1},\ldots ,p_{k+1})$ is a walk in $G^{\prime }$ between $%
p_{0}=s_{0}$ and $p_{k}=t_{0}$.

\begin{figure}[h]
\centering
\includegraphics[scale=0.7]{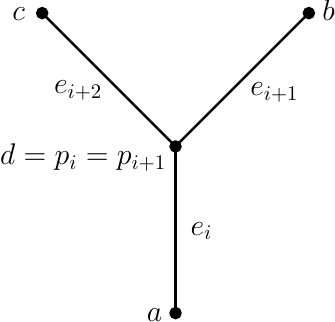}
\caption{A forbidden configuration.}
\label{shortest-respecting-forbidden-fig}
\end{figure}

Since $Q$ is a simple directed path, it follows that every edge of $J$
appears exactly once in $J$, and thus $J$ is a path of $G^{\prime }$. Now we
will prove that $J$ is actually a simple path of $G^{\prime }$. Suppose
otherwise that $p_{i}=p_{j}$ for some $0\leq i<j\leq k+1$. If $p_{j}=p_{k}$,
i.e.~$p_{j}=t_{0}$, then the subpath $(p_{0},p_{1},\ldots ,p_{i})$ of $J$
implies a strictly shorter directed path $Q^{\prime }$ than $Q$ between $%
s_{H}$ and $t_{H}$ in $H$, which is a contradiction. Therefore $p_{j}\neq
p_{k}$. Then, since $\lambda (p_{i-1}p_{i})< \lambda (p_{i}p_{i+1})$ for every $%
i=0,1,\ldots ,k-1$ by the construction of the directed graph $H$, it follows
in particular that $\lambda (p_{i-1}p_{i})< \lambda (p_{j}p_{j+1})$, and thus $%
\widehat{e_{i}e_{j+1}}$ is an arc in the directed graph $H$. Thus the path $%
(p_{0},p_{1},\ldots ,p_{i},p_{j+1},\ldots ,p_{k})$ of $G^{\prime }$ implies
a strictly shorter directed path $Q^{\prime }$ than $Q$ between $s_{H}$ and $%
t_{H}$ in $H$, which is again a contradiction. Therefore $p_{i}\neq p_{j}$
for every $0\leq i<j\leq k+1$ in $J$, and thus $J$ is a simple path in $%
G^{\prime }$ between $p_{0}=s_{0}$ and $p_{k}=t_{0}$. Finally, it is easy to
check that the weight $w(J)$ of $J$ in $G^{\prime }$ equals the weight $%
w_{H}(Q)$ of $Q$ in $H$.

Summarizing, there exists a journey $J$ in $G^{\prime }$
between $s_{0}$ and $t_{0}$ if and only if there is a directed path $Q$ in $%
H $ from $s_{H}$ to $t_{H}$. Furthermore, if such paths exist, then the
weight of the shortest journey $J$ of $G^{\prime }$ between $%
s_{0}$ and $t_{0}$ equals the weight of the shortest directed path $Q$ of 
$H$ from $s_{H}$ to $t_{H}$.

Moreover, the above proof immediately implies an efficient algorithm for
computing the graph $H$ from~$\lambda(G)$ (by first constructing the
auxiliary graph $G^{\prime }$ from $\lambda(G)$). This can be done in $O(\sum_{v\in
V}\delta_{v}^{2})$ time. Indeed, for every vertex $v$ of $G^{\prime }$ we add at
most $2{\binom{\delta_{v}}{2}}=\delta_{v}(\delta_{v}-1)$ arcs to $H$. That is, $|V_{H}|=m+2$
and $|E_{H}|\leq \sum_{v\in V(G^{\prime })}\delta_{v}(\delta_{v}-1)=O(\sum_{v\in
V}\delta_{v}^{2})$. After we construct $H$, we can compute a shortest directed
path between $s_{H}$ and $t_{H}$ in $O(|E_{H}|+|V_{H}|\log |V_{H}|)$ time
using Dijkstra's algorithm with Fibonacci heaps~\cite{FT87}. That
is, we can compute a shortest directed path $Q$ in $H$ between $s_{H}$ and $%
t_{H}$ in $O(m\log m+\sum_{v\in V}\delta_{v}^{2})$ time. Once we have computed
the path $Q$, we can easily construct the shortest undirected journey $J$ in $\lambda(G)$
between $s$ and $t$ in $O(m+n)$ time. This completes the proof of the
theorem.
\qquad
\end{proof}

\section{A Menger's Analogue for Temporal Graphs}
\label{sec:menger}

In \cite{KKK00}, Kempe \emph{et al.} proved that Menger's theorem, at least in its original formulation, does not hold for single-label temporal networks in which journeys must have non-decreasing labels (and not necessarily strictly increasing as in our case). For a counterexample, it is not hard to see in Figure \ref{fig:ber} that there are no two disjoint time-respecting paths from $v_1$ to $v_4$ but after deleting any one node (other than $v_1$ or $v_4$) there still remains a time-respecting $v_1$-$v_4$ path. Moreover, they proved that the violation of Menger's theorem in such temporal networks renders the computation of the number of disjoint $s$-$t$ paths NP-complete.

\begin{figure}[!hbtp]
   \centering{
        \includegraphics[width=0.5\textwidth]{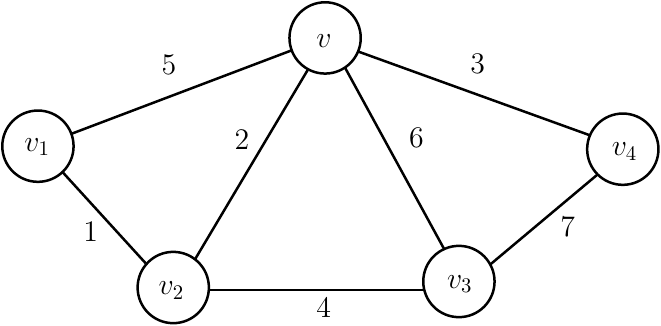}
        }
   \caption{A counterexample of Menger's theorem for temporal networks (adopted from \cite{KKK00}). Each edge has a single time-label indicating its availability time.} \label{fig:ber}
\end{figure}

We prove in this section that, in contrast to the above important negative result, there is a natural analogue of Menger's theorem that is valid for all temporal networks. In Theorem \ref{the:dmeng}, we define this analogue and prove its validity. Then as an illustration (\S \ref{subsec:application}), we show how using our theorem can simplify the proof of a recent token dissemination result.

When we say that we remove \emph{node departure time} $(u,t)$ we mean that we remove \emph{all time-edges leaving $u$ at time $t$}, i.e.~we remove label $t$ from all $(u,v)$ edges (for all $v\in V$). In case of an undirected graph, we replace each edge by two antiparallel edges and remove label $t$ only from the outgoing edges of $u$. So, when we ask how many node departure times are needed to separate two nodes $s$ and $v$ we mean how many node departure times must be selected so that after the removal of all the corresponding time-edges the resulting temporal graph has no $(s,v)$-journey. \footnote{Note that this is a different question from how many time-edges must be removed and, as we shall see, the latter question does not result in a Menger's analogue. Of course, removing a node departure time again results in the removal of some time-edges, but a Menger's analogue based on the number of those edges would not work. Instead, what turns out to work is an analogue based on counting the number of node departure times.}   

\begin{theorem} [Menger's Temporal Analogue] \label{the:dmeng}
Take any temporal graph $\lambda(G)$, where $G=(V,E)$, with two distinguished nodes $s$ and $v$. The maximum number of out-disjoint journeys from $s$ to $v$ is equal to the minimum number of node departure times needed to separate $s$ from $v$.  
\end{theorem}
\begin{proof}
Assume, in order to simplify notation, that $\lambda_{\min}=1$. Take the static expansion $H=(S,A)$ of $\lambda(G)$. Let $\{u_{i1}\}$ and $\{u_{in}\}$ represent $s$ and $v$ over time, respectively (first and last columns, respectively), where $0\leq i\leq \lambda_{\max}$. We extend $H$ as follows. For each $u_{ij}$, $0\leq i\leq \lambda_{\max}-1$, with at least 2 outgoing edges to nodes different than $u_{(i+1)j}$, e.g.~to nodes $u_{(i+1)j_1},u_{(i+1)j_2},\ldots,u_{(i+1)j_k}$, we add a new node $w_{ij}$ and the edges $(u_{ij},w_{ij})$ and $(w_{ij},u_{(i+1)j_1}),(w_{ij},u_{(i+1)j_2}),\ldots,(w_{ij},u_{(i+1)j_k})$. We also define an edge capacity function $c:A\rightarrow \{1,\lambda_{\max}\}$ as follows. All edges of the form $(u_{ij},u_{(i+1)j})$ take capacity $\lambda_{\max}$ and all other edges take capacity $1$. 
We are interested in the maximum flow from $u_{01}$ to $u_{\lambda_{\max}n}$. 
As this is simply a usual static flow network, the max-flow min-cut theorem applies stating that the maximum flow from $u_{01}$ to $u_{\lambda_{\max}n}$ is equal to the minimum of the capacity of a cut separating $u_{01}$ from $u_{\lambda_{\max}n}$. So it suffices to show that (i) the maximum number of out-disjoint journeys from $s$ to $v$ is equal to the maximum flow from $u_{01}$ to $u_{\lambda_{\max}n}$ and (ii) the minimum number of node departure times needed to separate $s$ from $v$ is equal to the minimum of the capacity of a cut separating $u_{01}$ from $u_{\lambda_{\max}n}$.

For (i) observe that any set of $h$ out-disjoint journeys from $s$ to $v$ corresponds to a set of $h$ disjoint paths from $u_{01}$ to $u_{\lambda_{\max}n}$ w.r.t. diagonal edges (edges in $E\bs\{(u_{ij},u_{(i+1)j})\}$) and inversely, so their maximums are equal. Next observe that any set of $h$ disjoint paths from $u_{01}$ to $u_{\lambda_{\max}n}$ w.r.t. diagonal edges corresponds to an integral $u_{01}$-$u_{\lambda_{\max}n}$ flow on $H$ of value $h$ and inversely. As the maximum integral $u_{01}$-$u_{\lambda_{\max}n}$ flow is equal to the maximum $u_{01}$-$u_{\lambda_{\max}n}$ flow (the capacities are integral and thus the integrality theorem of maximum flows applies) we conclude that the maximum $u_{01}$-$u_{\lambda_{\max}n}$ flow is equal to the maximum number of out-disjoint journeys from $s$ to $v$.

For (ii) observe that any set of $r$ node departure times that separate $s$ from $v$ corresponds to a set of $r$ diagonal edges leaving $u_{ij}$ nodes (ending either in $w_{ij}$ or in $u_{(i+1)j^\prime}$ nodes) that separate $u_{01}$ from $u_{\lambda_{\max}n}$ and inversely. Finally, observe that there is a minimum $u_{01}$-$u_{\lambda_{\max}n}$ cut on $H$ that only uses such edges: for if a minimum cut uses vertical edges we can replace them by diagonal edges and we can replace all edges leaving a $w_{ij}$ node by the edge $(u_{ij},w_{ij})$ without increasing the total capacity. 
\qquad
\end{proof}

\begin{corollary}
By symmetry we have that the maximum number of in-disjoint journeys from $s$ to $v$ is equal to the minimum number of node arrival times needed to separate $s$ from $v$.
\end{corollary}

\begin{corollary}
The following alternative statements are both valid:
\begin{itemize}
 \item The maximum number of time-node disjoint journeys from $s$ to $v$ is equal to the minimum number of time-nodes needed to separate $s$ from $v$.
 \item The maximum number of time-edge disjoint journeys from $s$ to $v$ is equal to the minimum number of time-edges needed to separate $s$ from $v$. \footnote{By time-node disjointness we mean that they do not meet on the same node at the same time (in terms of the expansion graph the corresponding paths should be disjoint in the classical sense) and by time-edge disjointness that they do not use the same time-edge (which again translates to using the same diagonal edge on the expansion graph).}
\end{itemize}
\end{corollary}

The following version is though violated: ``the maximum number of out-disjoint (or in-disjoint) journeys from $s$ to $v$ is equal to the minimum number of time-edges needed to separate $s$ from $v$'' (see Figure \ref{fig:mv-ex}). The same holds for the original statement of Menger's theorem as discussed in the beginning of this section (see \cite{KKK00}). 

\begin{figure}[!hbtp]
\centering{
\includegraphics[width=0.6\textwidth]{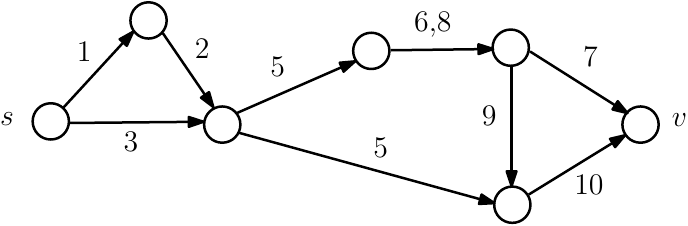}
}
\caption{A violation of an invalid Menger's analogue. Both edges labeled 5 must be removed to separate $s$ from $v$ however there are no two out-disjoint journeys from $s$ to $v$ (all $(s,v)$-journeys must use some edge labeled 5).} \label{fig:mv-ex}
\end{figure}

\subsection{An Application: Foremost Dissemination (Journey Packing)}
\label{subsec:application}

Consider the following problem. We are given a temporal graph $\lambda(G)$, where $G=(V,E)$, a source node $s$, a sink node $v$ and an integer $q$. We are asked to find the minimum arrival time of a set of $q$ out-disjoint $(s,v)$-journeys or even the minimizing set itself.

By exploiting the Menger's analogue proved in Theorem \ref{the:dmeng} (and in order to provide an example application of it), we give an alternative (and probably simpler to appreciate) proof of the following Lemma from \cite{DPRS13} (stated as Lemma \ref{lem:gathering} below) holding for a special case of temporal networks, namely those that have \emph{connected instances}. Formally, a temporal network $\lambda(G)$ is said to have connected instances if $\lambda(G,t)$ is connected at all times $t\in\bbbn$. The problem under consideration is distributed $k$-token dissemination: there are $k$ tokens assigned to some given source nodes. In each round (i.e.~discrete moment in the temporal network), each node selects a single token to be sent to all of its current neighbors (i.e.~broadcast). The current neighbors at round $i$ are those defined by $E(i)$. The goal of a distributed protocol (or of a centralized strategy for the same problem) is to deliver all tokens to a given sink node $v$ as fast as possible. We assume that the algorithms know the temporal network in advance.

\begin{lemma} \label{lem:gathering}
Let there be $k\leq n$ tokens at given source nodes and let $v$ be an arbitrary node. Then, all the tokens can be sent to $v$ using broadcasts in $O(n)$ rounds.
\end{lemma}

Let $S=\{s_1,s_2,\ldots,s_h\}$ be the set of source nodes and let $k(s_i)$ be the number of tokens of source node $s_i$, so that $\sum_{1\leq i\leq h} k(s_i)=k$. Clearly, it suffices to prove the following lemma.

\begin{lemma} \label{lem:tok}
We are given a temporal graph $\lambda(G)$ with connected instances and age $\alpha(\lambda)=n+k$. 
We are also given a set of source nodes $S\subseteq V$, a mapping $k:S\rightarrow \bbbn_{\geq 1}$ so that $\sum_{s\in S} k(s)=k$, and a sink node $v$. Then there are at least $k$ out-disjoint journeys from $S$ to $v$ such that $k(s_i)$ journeys leave from each source node $s_i$.
\end{lemma}
\begin{proof}
We conceive $k(s)$ as the number of tokens of source $s$. Number the tokens arbitrarily. Create a supersource node $s^\prime$ and connect it to the source node with token $i$ by an edge labeled $i$. Increase all other edge labels by $k$. Clearly the new temporal graph $D=\lambda^\prime(G^\prime)$ has asymptotically the same age as the original and all properties have been preserved (we just shifted the original temporal graph in the time dimension). Moreover, if there are $k$ out-disjoint journeys from $s^\prime$ to $v$ in $D$ then by construction of the edges leaving $s^\prime$ we have that precisely $k(s)$ of these journeys must be leaving from each source $s\in S$. So it suffices to show that there are $k$ out-disjoint journeys from $s^\prime$ to $v$. By Theorem \ref{the:dmeng} it is equivalent to show that the minimum number of departure times that must be removed from $D$ to separate $s^\prime$ from $v$ is $k$. Assume that we remove $y<k$ departure times. Then for more than $n$ rounds all departure times are available (as we have $n+2k$ rounds and we just have $y<k$ removals). As every instance of $G$ is connected, we have that there is always an edge in the cut between the nodes that have been reached by $s^\prime$ already and those that have not, unless we remove some departure times. As for more than $n$ rounds all departure times are available it is immediate to observe that $s^\prime$ reaches $v$ implying that we cannot separate $s^\prime$ from $v$ with less that $k$ removals and this completes the proof.
\qquad  
\end{proof}

\part*{Part II}

\section{Minimum Cost Temporal Connectivity}
\label{min-cost-connectivity-sec}

In this section, we introduce some cost measures for maintaining different
types of temporal connectivity. According to these temporal connectivity
types, individuals are required to be capable to communicate with other
individuals over the dynamic network, possibly with further restrictions on
the timing of these connections. We initiate this study by considering the
following fundamental problem: Given a (di)graph $G$, assign labels to the
edges of $G$ so that the resulting temporal graph $\lambda (G)$ minimizes
some parameter and at the same time preserves some connectivity property of $%
G$ in the time dimension. For a simple illustration of this, consider the
case in which $\lambda (G)$ should contain a journey from $u$ to $v$ if and
only if there exists a path from $u$ to $v$ in $G$. In this example, the
reachabilities of $G$ completely define the temporal reachabilities that $%
\lambda (G)$ is required to have.

We consider two cost optimization criteria for a (di)graph $G$. The first
one, called \emph{temporality} of $G$, measures the maximum number of labels
that an edge of $G$ has been assigned. The second one, called \emph{temporal
cost} of $G$, measures the total number of labels that have been assigned to
all edges of $G$. That is, if we interpret the number of assigned labels as
a measure of \emph{cost}, the temporality (resp.~the temporal cost)\ of $G$
is a measure of the decentralized (resp.~centralized) cost of the network,
where only the cost of individual edges (resp.~the total cost over all
edges) is considered. We introduce these cost parameters in Definition~\ref%
{temporality-cost-def}. Each of these two cost measures can be minimized
subject to some particular connectivity property $\mathcal{P}$ that the
labeled graph $\lambda (G)$ has to satisfy. For simplicity of notation, we
consider in Definition~\ref{temporality-cost-def} the connectivity property $%
\mathcal{P}$ as a subset of the set $\mathcal{L}_{G}$ of all possible
labelings $\lambda $ on the (di)graph $G$. Furthermore, the minimization of
each of these two cost measures can be affected by some problem-specific
constraints on the labels that we are allowed to use. We consider here one
of the most natural constraints, namely an upper bound on the \emph{age} of
the constructed labeling $\lambda $.

\begin{definition}
\label{temporality-cost-def}Let $G=(V,E)$ be a (di)graph, $\alpha _{\max}\in \mathbb{N}$, and $\mathcal{P}$ be a connectivity property. 
Then the \emph{temporality} of $(G,\mathcal{P},\alpha _{\max })$ is%
\begin{equation*}
\tau (G,\mathcal{P},\alpha _{\max })=\min_{\lambda \in \mathcal{P\cap L}%
_{G,\alpha _{\max }}}\max_{e\in E}|\lambda (e)|
\end{equation*}%
and the \emph{temporal cost} of $(G,\mathcal{P},\alpha _{\max })$ is%
\begin{equation*}
\kappa (G,\mathcal{P},\alpha _{\max })=\min_{\lambda \in \mathcal{P\cap L}%
_{G,\alpha _{\max }}}\sum_{e\in E}|\lambda (e)|
\end{equation*}%
Furthermore $\tau (G,\mathcal{P})=\tau (G,\mathcal{P},\infty )$ and $\kappa
(G,\mathcal{P})=\kappa (G,\mathcal{P},\infty )$.
\end{definition}

Note that Definition~\ref{temporality-cost-def} can be stated for an
arbitrary property $\mathcal{P}$ of the labeled graph $\lambda (G)$ (e.g.
some proper coloring-preserving property). Nevertheless, we only consider
here $\mathcal{P}$ to be a connectivity property of $\lambda (G)$. In
particular, we investigate the following two connectivity properties $%
\mathcal{P}$:

\begin{itemize}
\item \emph{all-paths}$(G)=\{\lambda \in \mathcal{L}_{G}:$ for all simple
paths $P$ of $G$, $\lambda $ preserves $P\}$,

\item \emph{reach}$(G)=\{\lambda \in \mathcal{L}_{G}:$ for all $u,v\in V$
where $v$ is reachable from $u$ in $G$, $\lambda$ preserves at least one
simple path from $u$ to $v\}$.
\end{itemize}

\subsection{Basic Properties of Temporality Parameters}
\label{basic-properties-cost-subsec}

\subsubsection{Preserving All Paths}

We begin with some simple observations on $\tau(G,\ap)$. Recall that given a (di)graph $G$ our goal is to label $G$ so that all simple paths of $G$ are preserved by using as few labels per edge as possible. From now on, when we say ``graph'' we will mean a directed one and we will state it explicitly when our focus is on undirected graphs.

Another interesting observation is that if $p(G)$ is the length of the longest path in $G$ then we can trivially preserve all paths of $G$ by using $p(G)$ labels per edge. Give to every edge the labels $\{1,2,\ldots,p(G)\}$ and observe that for every path $e_1,e_2,\ldots,e_k$ of $G$ we can use the increasing sequence of labels $1,2,\ldots,k$ due to the fact that $k\leq p(G)$. Thus, we conclude that the upper bound $\tau(G,\ap)\leq p(G)$ holds for all graphs $G$. Of course, note that equality is easily violated. For example, a directed line has $p(G)=n$ but $\tau(G,\ap)=1$.

\begin{observation}
$\tau(G,\text{all paths})\leq p(G)$ for all graphs $G$.
\end{observation}

\noindent \textbf{Directed Rings.} The following proposition states that if $G$ is a directed ring then the temporality of preserving all paths is 2. This means that the minimum number of labels per edge that preserve all simple paths of a ring is 2. As the proof was already sketched in Section \ref{sec:intro}, we don't provide a proof here.

\begin{proposition} \label{pro:ring}
$\tau(G,\text{all paths})=2$ when $G$ is a ring and $\tau(G,\text{all paths})\geq 2$ when $G$ contains a ring.
\end{proposition}

\noindent \textbf{Directed Acyclic Graphs.} A topological sort of a digraph $G$ is a linear ordering of its nodes such that if $G$ contains an edge $(u,v)$ then $u$ appears before $v$ in the ordering. It is well known that a digraph $G$ can be topologically sorted iff it has no directed cycles that is iff it is a DAG. A topological sort of a graph can be seen as placing the nodes on a horizontal line in such a way that all edges go from left to right; see e.g.~\cite[page 549]{CLRS01}.

\begin{proposition} \label{pro:dag-all-paths}
If $G$ is a DAG then $\tau(G,\text{all paths})=1$.
\end{proposition}
\begin{proof}
Take a topological sort $u_1,u_2,\ldots,u_n$ of $G$. Clearly, every edge is of the form $(u_i,u_j)$ where $i<j$. Give to every edge $(u_i,u_j)$ label $i$, that is $\lambda(u_i,u_j)=i$ for all $(u_i,u_j)\in E$. Now take any node $u_l$. Each of its incoming edges has some label $l^\prime<l$ and all its outgoing edges have label $l$. Now take any simple path $p=v_1,v_2,\ldots,v_k$ of $G$. Clearly, $v_i$ appears before $v_{i+1}$ in the topological sort for all $1\leq i\leq k-1$, which implies that $\lambda(v_i,v_{i+1})<\lambda(v_{i+1},v_{i+2})$, for all $1\leq i\leq k-2$. This proves that $p$ is preserved. As we have preserved all simple paths with a single label on every edge, we conclude that $\tau(G,\ap)=1$ as required.
\qquad
\end{proof}

\subsubsection{Preserving All Reachabilities}

Now, instead of preserving all paths, we impose the apparently simpler requirement of preserving just a single path between every reachability pair $u,v\in V$. We claim that it is sufficient to understand how $\tau(G,\reach)$, behaves on strongly connected digraphs. Let $\cc(G)$ be the set of all strongly connected components of a digraph $G$. The following lemma proves that, w.r.t. the $\reach$ property, the temporality of any digraph $G$ is equal to the maximum temporality of its components.

\begin{lemma} \label{lem:components}
$\tau(G,\reach) = \max\{1, \max_{C\in\cc(G)} \tau(C,\reach)\}$ for every digraph $G$ with at least one edge. In the case of no edge, $\tau(G,\reach) = 0$ trivially.
\end{lemma}
\begin{proof}
Take any digraph $G$. Now take the DAG $D$ of the strongly connected components of $G$. The nodes of $D$ are the components of $G$ and there is an edge from component $C$ to component $C^\prime$ if there is an edge in $G$ from some node of $C$ to some node of $C^\prime$. As $D$ is a DAG, we can obtain a topological sort of it which is a labeling $C_1,C_2,\ldots ,C_t$ of the $t$ components so that all edges between components go only from left to right.

In the case where at least one component has at least 2 nodes (in which case $\max_{C\in\cc(G)} \tau(C,\reach)\geq 1$), we have to prove that we can label $G$ by using at most $\max_{1\leq i\leq t} \tau(C_i,\reach)$ labels per edge and that we cannot do better than this. Consider the following labeling process. 
For each component $C_i$ define $d_i= \min_{\lambda\in\cc_i}(\lambda_{\max}(\lambda) - \lambda_{\min}(\lambda))$, where $\cc_i$ is the set of all labelings of $C_i$ that preserve all of its reachabilities using at most $\tau(C_i,\reach)$ labels per edge. Note that any $C_i$ can be labeled beginning from any desirable $\lambda_{\min}$ with at most $\tau(C_i,\reach)$ labels per edge and with $\lambda_{max}$ equal to $\lambda_{\min}+d_i$. Now, label component $C_1$ with $\lambda_{\min}=1$ and $\lambda_{\max}=1+d_1$. Label all edges leaving $C_1$ with label $d_1+2$. Label component $C_2$ with $\lambda_{\min}=d_1+3$ and $\lambda_{\max}=(d_1+3)+d_2$ and all its outgoing edges with label $(d_1+3)+d_2+1$. In general, label component $C_i$ with $\lambda_{\min}=1+\sum_{1\leq j\leq i-1} (d_j+2)$ and $\lambda_{\max}=\lambda_{\min}+d_i$ and label all edges leaving $C_i$ with label $\lambda_{\max}+1$. It is not hard to see that this labeling scheme preserves all reachabilities of $G$ using just one label on each edge of $G$ corresponding to an edge of $D$ and at most $\tau(C_i,\reach)$ labels per edge inside each component $C_i$. Thus, it uses at most $\max_{1\leq i\leq t} \tau(C_i,\reach)$ labels on every edge. By observing that for each strongly connected component $C_i$, $\tau(C_i,\reach)$ must be paid by any labeling of $G$ that preserves all reachabilities in that component, the equality $\tau(G,\reach) = \max_{C\in\cc(G)} \tau(C,\reach)$ follows. 

In the extreme case where all components are just single nodes (in which case $\max_{C\in\cc(G)} \tau(C,\reach)=0$), it holds that $D=G$, therefore $G$ itself is a DAG and we only need 1 label per edge (as in Proposition \ref{pro:dag-all-paths}) and, thus, $\tau(G,\reach) = 1$. 
\qquad
\end{proof}

Lemma \ref{lem:components} implies that any upper bound on the temporality of preserving the reachabilities of strongly connected digraphs can be used as an upper bound on the temporality of preserving the reachabilities of general digraphs. In view of this, we focus on strongly connected digraphs $G$. 

We begin with a few simple but helpful observations. Obviously, $\tau(G,\reach)\leq\tau(G,\ap)$ as any labeling that preserves all paths trivially preserves all reachabilities as well. If $G$ is a clique then $\tau(G,\reach) = 1$ as giving to each edge a single arbitrary label (e.g.~label 1 to all) preserves all direct connections (one-step reachabilities) which are all present. 
If $G$ is a directed ring (which is again strongly connected) then it is easy to see that $\tau(G,\reach) = 2$.  
An interesting question is whether there is some bound on $\tau(G,\reach)$ either for all digraphs or for specific families of digraphs. 
The following lemma proves that indeed there is a very satisfactory generic upper bound.

\begin{lemma} \label{lem:strongly-connected}
$\tau(G,\reach)\leq 2$ for all strongly connected digraphs $G$.
\end{lemma}
\begin{proof}
As $G$ is strongly connected, if we pick any node $u$ then for all $v$ there is a $(v,u)$ and a $(u,v)$-path. As for any $v$ there is a $(v,u)$-path, then we may form an in-tree $T_{in}$ rooted at $u$ (that is a tree with all directions going upwards to $u$). Now beginning from the leaves give any direction preserving labeling (just begin from labels 1 at the leaves and increase them as you move upwards). Say that the depth is $k$ which means that you have increased up to label $k$. Now consider an out-tree $T_{out}$ rooted at $u$ that has all edge directions going from $u$ to the leaves. To make things simpler create second copies of all nodes but $u$ so that the two trees are disjoint (w.r.t. to all nodes but $u$). In fact, one tree passes through all the first copies and arrives at $u$ and the other tree begins from $u$ and goes to all the second copies. Now we can begin the labeling of $T_{out}$ from $k+1$ increasing labels as we move away from $u$ on $T_{out}$. This completes the construction.

Now take any two nodes $w$ and $v$. Clearly, there is a time-respecting path from $w$ to $u$ and then a time-respecting path from $u$ to $v$ using greater labels so there is a time-respecting path from $w$ to $v$. Finally, notice that for any edge on $T_{in}$ there is at most one copy of that edge on $T_{out}$ thus clearly we use at most 2 labels per edge. 
\qquad
\end{proof}

Combining Lemma \ref{lem:components} and Lemma \ref{lem:strongly-connected} gives the following theorem:
\begin{theorem}
$\tau(G,\reach)\leq 2$ for all digraphs $G$.
\end{theorem}

\subsubsection{Restricting the Age}

Now notice that for all $G$ we have $\tau(G, \reach,$ $d(G))\leq d(G)$; recall that $d(G)$ denotes the diameter of (di)graph $G$. Indeed it suffices to label each edge by $\{1,2,\ldots,d(G)\}$. Since every shortest path between two nodes has length at most $d(G)$, in this manner we preserve all shortest paths and thus all reachabilitities arriving always at most by time $d(G)$, thus we also preserve the diameter. Thus, a clique $G$ has trivially $\tau(G, \reach, d(G)) = 1$ as $d(G)=1$ and we can only have large $\tau(G, \reach, d(G))$ in graphs with large diameter. For example, a directed ring $G$ of size $n$ has $\tau(G, \reach, d(G))=n-1$ (note that on a ring it always holds that $\tau(G, \reach, k)=\tau(G, \ap, k)$, as on a ring it happens that satisfying all reachabilities also satisfies all paths while the inverse is true for all graphs). Indeed, assume that from some edge $e$, label $1\leq i \leq n-1$ is missing. It is easy to see that there is some shortest path between two nodes of the ring that in order to arrive by time $n-1$ must use edge $e$ at time $i$. As this label is missing, it uses label $i+1$, thus it arrives by time $n$ which is greater than the diameter. In this particular example we can preserve the diameter only if all edges have the labels $\{1,2, \ldots,n-1\}$.

On the other hand, there are graphs with large diameter in which $\tau(G,\reach,d(G))$ is small. This may also be the case even if G is strongly connected. For example, consider the graph with nodes $u_1,u_2,\ldots,u_n$ and edges $(u_i, u_{i+1})$ and $(u_{i+1}, u_i)$ for all $1\leq i \leq n-1$. In words, we have a directed line from $u_1$ to $u_n$ and an inverse one from $u_n$ to $u_1$. The diameter here is $n-1$ (e.g.~the shortest path from $u_1$ to $u_n$). On the other hand, we have $\tau(G,\reach,d(G)) = 1$: simply label one path $1,2,..., n-1$ and label the inverse one $1,2,...,n-1$ again, i.e.~give to edges $(u_i, u_{i+1})$ and $(u_{n-i+1}, u_{n-i+2})$ label $i$. The reason here is that there are only two pairs of nodes that must necessarily use the long paths $(u_1,u_n)$ and $(u_n, u_1)$ and preserve the diameter $n-1$. All other smaller shortest paths between other pairs of nodes have now a big gap of $n-1$ to exploit.

We will now demonstrate what makes $\tau(G,\reach,d(G))$ grow. It happens when many maximum shortest paths (those that determine the diameter of $G$) between different pairs of nodes that are additionally unique (the paths), in the sense that we must necessarily take them in order to preserve the reachabilities (it may hold even if
they are not unique but this simplifies the argument), all pass through the same edge $e$ but use $e$ at many different times. It will be helpful to look at Figure \ref{fig:diam}. Each $(u_i,v_i)$-path is a unique shortest path between $u_i$ and $v_i$ and has additionally length equal to the diameter (i.e.~it is also a maximum one), so we must necessarily preserve all 5 $(u_i,v_i)$-paths. Note now that each $(u_i,v_i)$-path passes through $e=(u_1,v_5)$ via its $i$-th edge. Each of these paths can only be preserved without violating $d(G)$ by assigning the labels $1,2,\ldots, d(G)$, however note that then edge $e$ must necessarily have all labels $1,2,\ldots,d(G)$. To see this, notice simply that if any label $i$ is missing from $e$ then there is some maximum shortest path that goes through $e$ at step $i$. As $i$ is missing it cannot arrive sooner than time $d(G)+1$ which violates the preservation of the diameter.

\begin{figure}[!hbtp]
\centering{
\includegraphics[width=0.8\textwidth]{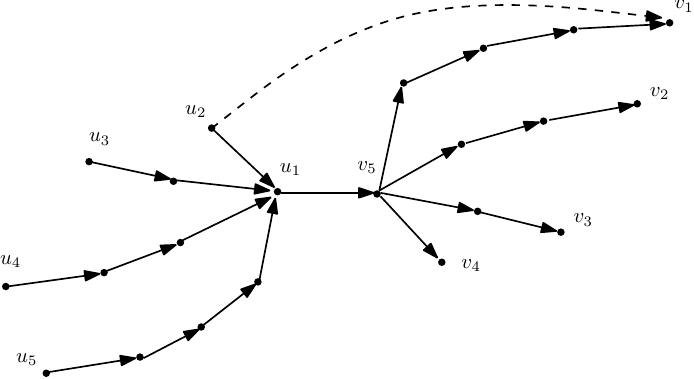}
}
\caption{An example graph in which $\tau(G, \reach, d(G)) = d(G)$. All paths longer than length 5 that are formed are not shortest paths, e.g.~there is a path (the dashed one) of length at most 5 from $u_2$ to $v_1$ and the same for all other such pairs.} \label{fig:diam}
\end{figure} 

\noindent\textbf{Undirected Tree.} Now consider an undirected tree $T$. 

\begin{corollary}
If $T$ is an undirected tree then $\tau(T,\text{all paths},d(T))\leq 2$.
\end{corollary}
\begin{proof}
This follows as a simple corollary of Lemma \ref{lem:strongly-connected}. If we replace each undirected edge by two antiparallel edges, then $T$ is a strongly connected digraph and, additionally, for every ordered pair of nodes $(u,v)$ there is precisely one simple path from $u$ to $v$. The latter implies that preserving all paths of $T$ is equivalent to preserving all reachabilities of $T$. So, all assumptions of Lemma \ref{lem:strongly-connected} are satisfied and therefore $\tau(T,\text{all paths})\leq 2$. Finally, recall that the labeling of the construction in the proof of Lemma \ref{lem:strongly-connected} starts increasing labels level-by-level from the leaves to the root and then from the root to the leaves, therefore the number of increments (i.e., the maximum label used) is upper bounded by the diameter of $T$, thus, $\tau(T,\text{all paths},d(T))\leq 2$ as required.  
\qquad
\end{proof}
\newline

\noindent\textbf{Trade-off on a Ring.} We shall now prove that there is a trade-off between the temporality and the age. In particular, we consider a directed ring $G=(e_1,e_2,\ldots,e_n)$, where the $e_i$ are edges oriented clockwise. As we have already discussed, if $\alpha= n-1$ then $\tau(G,\ap,\alpha)=n-1$ (which is the worst possible) and if $\alpha= 2(n-1)$ then $\tau(G,\ap,\alpha)=2$ (which is the best possible). We now formalize the behavior of $\tau$ as $\alpha$ moves from $n-1$ to $2(n-1)$.

\begin{theorem} \label{the:tradeoff}
If $G$ is a directed ring and $\alpha=(n-1)+k$, where $1\leq k \leq n-1$, then $\tau(G,\text{all paths},\alpha)=\Theta(n/k)$ and in particular $\lfloor\frac{n-1}{k+1}\rfloor\leq\tau(G,\text{all paths},\alpha)\leq\lceil\frac{n}{k+1}\rceil+1$. Moreover, $\tau(G,\text{all paths},n-1)=n-1$ (i.e.~when $k=0$).
\end{theorem}
\begin{proof}
The proof of the upper bound is constructive. In particular, we present a labeling that preserves all paths of the ring $G$ using at most $\lceil\frac{n}{k+1}\rceil+1$ labels on every edge and maximum label $(n-1)+k$. Let the ring be $e_1,e_2,\ldots,e_n$ and clockwise. We say that an edge $e_i$ is \emph{satisfied} if there is a journey of length $n-1$ beginning from $e_i$ (clearly, considering only those journeys that do not use a label greater than $\alpha=(n-1)+k$). Consider the following labeling procedure.
\begin{itemize}
 \item For all $i=0,1,2,\ldots,\lceil\frac{n}{k+1}\rceil-2$
  \begin{itemize}
   \item Assign label 1 to edge $e_{j=i(k+1)+1}$.
   \item Beginning from edge $e_{j+1}$, assign labels $2,3,\ldots,(n-1)+k$ clockwise.
  \end{itemize}
 \item For $i=\lceil\frac{n}{k+1}\rceil-1$, assign label 1 to edge $e_{j=i(k+1)+1}$ and beginning from edge $e_{j+1}$ assign labels $2,3,\ldots,(n-1)+(n-j)$ clockwise.
\end{itemize}
Note that in each iteration $i$ we satisfy edges $e_{i(k+1)+1},e_{i(k+1)+2},\ldots,e_{(i+1)(k+1)}$, i.e.~$k+1$ new edges, without leaving gaps. It follows that in $\lceil\frac{n}{k+1}\rceil$ iterations all edges have been  satisfied. The first iteration assigns at most two labels on edge $e_1$ and every other iteration, apart from the last one, assigns one label on $e_1$ (and clearly at most one on every other edge), thus $e_1$ gets a total of at most $\lceil\frac{n}{k+1}\rceil+1$ labels (and all other edges get at most this).

Now, for the lower bound, take an arbitrary edge, e.g.~$e_1$. Given an edge $e_i$ and a journey $J$ from $e_i$ to $e_1$ that uses label $l_1$ on $e_1$, define the delay of $J$ as $l_1-l(J)$, where $l(J)$ is the length of journey $J$ i.e.~$n-i+2$. In words, the delay of a $(e_i,e_1)$-journey is the difference between the time at which the journey visits $e_1$ minus the fastest time that it could have visited $e_1$. Now, beginning from $e_n$ count $k+1$ times counterclockwise, i.e.~consider edge $e_{n-k}$. We show that in order to satisfy $e_{n-k}$ we must necessarily use one of the labels $\{k+2,k+3,\ldots,2k+2\}$ on $e_1$. To this end, notice that the delay of any journey that satisfies some edge can be at most $k$, the reason being that a delay of $k+1$ or greater implies that the journey cannot visit $n-1$ edges in less than $(n-1)+(k+1)$ time, thus it will have to use some label greater than $\alpha=(n-1)+k$, which is the maximum allowed. Thus, the maximum label by which a journey that satisfies $e_{n-k}$ can go through $e_1$ is $l(e_{n-k})+k=2k+2$, where $l(e_{i})$ denotes the length of the path beginning from the tail of $e_i$ and ending at the head of $e_1$. Moreover, the minimum label by which any journey from $e_{n-k}$ can go through $e_1$ is $l(e_{n-k})=k+2$. Thus, we conclude that any journey that satisfies $e_{n-k}$ has to use one of the labels $\{k+2,k+3,\ldots,2k+2\}$ on $e_1$.

It is not hard to see that the above idea generalizes as follows. For all $i=0,1,\ldots,\lfloor\frac{n-1}{k+1}\rfloor-1$, in order to satisfy edge $e_{n-i(k+1)+1}$ (note that $e_{n+1}=e_1$) we must necessarily use one of the labels $\{i(k+1)+1, i(k+1)+2,\ldots,(i+1)(k+1)\}$ on $e_1$. For example, for $i=0$ we get $\{1,2,\ldots,k+1\}$, for $i=1$ we get $\{k+2,\ldots,2k+2\}$, for $i=2$ we get $\{2k+3,\ldots,3k+3\}$, and so on. In summary, as the above sets are disjoint, if we begin from $e_1$ and move counterclockwise then for every $k+1$ edges we encounter we must pay for another (new) label on $e_1$ thus we pay at least $\lfloor\frac{n-1}{k+1}\rfloor$.
\qquad
\end{proof}

\subsection{A Generic Method for Computing Lower Bounds for Temporality}
\label{generic-method-subsec}

Proposition \ref{pro:ring} showed that graphs with directed cycles need at least 2 labels on some edge(s) in order for all paths to be preserved. Now a natural question to ask is whether we can preserve all paths of any graph by using at most 2 labels (i.e.~whether $\tau(G,\ap)\leq 2$ holds for all graphs). We shall prove that there are graphs $G$ for which $\tau(G,\ap)=\Omega(p(G))$ (recall that $p(G)$ denotes the length of the longest path in $G$), that is graphs in which the optimum labeling, w.r.t. temporality, is very close to the trivial labeling $\lambda(e)=\{1,2,\ldots,p(G)\}$, for all $e\in E$, that always preserves all paths. 

\begin{definition} \label{def:kernel}
Call a set $K=\{e_1,e_2,\ldots,e_k\}\subseteq E(G)$ of 
edges of a digraph $G$ an \emph{edge-kernel} if for every permutation $\pi=(e_{i_1},e_{i_2},\ldots,e_{i_k})$ of the elements of $K$ there is a simple path $P$ of $G$ that visits all edges of $K$ in the ordering defined by the permutation $\pi$.
\end{definition}

We will now prove that an edge-kernel of size $k$ needs at least $k$ labels on some edges. 
Our proof is constructive. In particular, given any labeling using $k-1$ labels on an edge-kernel of size $k$, we present a specific path that forces a $k$th label to appear. 

\begin{theorem} [Edge-kernel Lower Bound] \label{the:kernel}
If a digraph $G$ contains an edge-kernel of size $k$ then $\tau(G,\text{all paths})\geq k$.
\end{theorem}
\begin{proof}
Let $K=\{e_1,e_2,\ldots,e_k\}$ be such an edge-kernel of size $k$. Assume for contradiction that there is a path-preserving labeling using on every edge at most $k-1$ labels. Then there is a path-preserving labeling that uses precisely $k-1$ labels on every edge (just extend the previous labeling by arbitrary labels). On every edge $e_i$, $1\leq i\leq k$, sort the labels in an ascending order and denote by $\lambda_l(e)$ the $l$th smallest label of edge $e$; e.g.~if an edge $e$ has labels $\{1,3,7\}$, then $\lambda_1(e)=1$, $\lambda_2(e)=3$, and $\lambda_3(e)=7$. Note that, by definition of an edge-kernel, all possible permutations of the edges in $K$ appear in paths of $G$ that should be preserved. We construct a permutation $\pi=(e_{j_1},e_{j_2},\ldots,e_{j_k})$ of the edges in $K$ which cannot be time-respecting without using a $k$th label on some edge. As $e_{j_1}$ use the edge with the maximum $\lambda_1$, that is $\arg\max_{e\in K} \lambda_1(e)$. Then as $e_{j_2}$ use the edge with the maximum $\lambda_2$ between the remaining edges, that is $\arg\max_{e\in K\bs\{e_{j_1}\}} \lambda_2(e)$, and define $e_{j_3}, e_{j_4}, \ldots$ analogously. It is not hard to see that $\pi$ satisfies $\lambda_i(e_{j_i})\geq \lambda_i(e_{j_{i+1}})$ for all $1\leq i\leq k-1$. This, in turn, implies that for $\pi$ to be time-respecting it cannot use the labels $\lambda_1,\ldots,\lambda_{i-1}$ at edge $e_{j_i}$, for all $i\geq 2$, which shows that at edge $e_{j_k}$ it can use none of the $k-1$ available labels, thus a $k$th label is necessarily needed and the theorem follows.
\qquad
\end{proof}

\begin{lemma}
If $G$ is a complete digraph of order $n$ then it has an edge-kernel of size $\lfloor n/2\rfloor$.
\end{lemma}
\begin{proof}
Note that $\lfloor n/2\rfloor$ is the size of a maximum matching $M$ of $G$. As all possible edges that connect the endpoints of the edges in $M$ are available, $M$ is an edge-kernel of size $\lfloor n/2\rfloor$.
\qquad
\end{proof}

Now, Theorem \ref{the:kernel} implies that a complete digraph of order $n$ requires at least $\lfloor n/2\rfloor$ labels on some edge in order for all paths to be preserved, that is $\lfloor n/2\rfloor\leq \tau(G,\ap)$. At the same time we have the trivial upper bound $\tau(G,\ap)\leq n-1$ which follows from the fact that the longest path of a clique is hamiltonian, thus has $n-1$ edges, and for any graph $G$ the length of its longest path is an upper bound on $\tau(G,\ap)$.

The above, clearly remain true for the following (close to complete) bipartite digraph. There are two partitions $A=\{u_i:1\leq i\leq k\}$ and $B=\{v_i:1\leq i\leq k\}$ both of size $k$. The edge set consists of $(u_i,v_i)$ for all $i$ and $(v_i,u_j)$ for all $i,j$. In words, from $A$ to $B$ we have only horizontal connections while from $B$ to $A$ we have all possible connections.

\begin{lemma}
\label{planar-lower-bound-edge-kernel-lem}There exist planar graphs $G$ with 
$n$ vertices having edge-kernels of size $\Omega (n^{\frac{1}{3}})$.
\end{lemma}

\begin{proof}
The proof is done by construction. Consider the grid graph $G=G_{2n^{2},2n}$%
, i.e.~$G$ is formed as a part of the infinite grid having width of $2n^{2}$
vertices and height of $2n$ vertices. Note that $G$ is a planar graph. For
simplicity of the presentation, we consider the grid graph $G$ on the
Euclidean plane, where the vertices have integer coordinates and the lower
left vertex has coordinates $(1,1)$. Furthermore denote by $v_{i,j}$ the
vertex of $G$ that is placed on the point $(i,j)$, where $1\leq i\leq 2n^{2}$
and $1\leq j\leq 2n$. 
For every $i\in \{1,2,\ldots ,n\}$ denote 
$p_{i}=v_{(2i-1)n,n}$ and $q_{i}=v_{(2i-1)n+1,n}$. 
We define the edge subset $S=\{e_{i}=p_{i}q_{i}:1\leq i\leq n\}$.

We now prove that $S$ is an edge-kernel of $G$. Let $\pi
=(e_{i_{1}},e_{i_{2}},\ldots ,e_{i_{n}})$ be an arbitrary permutation of the
edges of $S=\{e_{1},e_{2},\ldots ,e_{n}\}$. We construct a simple path $P$
in $G$ that visits all the edges of $S$ in the order of the permutation $\pi 
$. That is, we construct a path $%
P=(p_{i_{1}},q_{i_{1}},P_{1},p_{i_{2}},q_{i_{2}},P_{2},\ldots
p_{i_{n-1}},q_{i_{n-1}},P_{n-1},p_{i_{n}},q_{i_{n}})$. In order to do so, it
suffices to define iteratively the simple paths $P_{1},P_{2},\ldots ,P_{n-1}$
such that no two of these paths share a common vertex. The path $P_{1}$
starts at $q_{i_{1}}$ and continues upwards on the column of $q_{i_{1}}$ in
the grid, until it reaches the top $2n$th row of the grid. Then, if $%
i_{2}>i_{1}$ (resp.~if $i_{2}<i_{1}$), the path $P_{1}$ continues on this
top row to the right (resp.~to the left), until it reaches the column of
vertex $p_{i_{2}}$ of the grid. Finally it continues downwards on this
column until it reaches $p_{i_{2}}$, where $P_{1}$ ends.

Consider now an index $t\in \{2,3,\ldots ,n-1\}$. In a similar manner as $%
P_{1}$, the path $P_{t}$ starts at vertex $q_{i_{t}}$. Then it continues
upwards on the column of $q_{i_{t}}$ in the grid as much as possible, such
that it does not reach any vertex of a path $P_{k}$, where $k\leq t-1$. Note
that, if no path $P_{k}$, $k\leq t-1$, passes through any vertex of the
column of $q_{i_{t}}$ in the grid, then the path $P_{t}$ reaches the top $2n$%
th row of the grid in this column. On the other hand, note that, since $%
q_{i_{t}}=v_{(2i_{t}-1)n+1,n}$ and $t\leq n-1$, at most the upper $t-1\leq
n-2$ vertices of the column of $q_{i_{t}}$ in the grid can possibly belong
to a path $P_{k}$, where $k\leq t-1$. Thus the path $P_{t}$ can always
continue upwards from $q_{i_{t}}$ by at least one edge. Let $a_{t}$ be the
uppermost vertex of $P_{t}$ on the column of $q_{i_{t}}$ of the grid (cf.
Figure~\ref{planar-lower-bound-fig} for $t=5$ and $e_{i_{5}}=e_{1}$).

Assume that $i_{t+1}>i_{t}$, i.e.~vertex $p_{i_{t+1}}$ lies to the right of
vertex $q_{i_{t}}$ on the $n$th row of the grid. Then, the path $P_{t}$
continues from vertex $a_{t}$ to the right, as follows. If $P_{t}$ can reach
the column of $p_{i_{t}}$ without passing through a vertex of a path $P_{k}$, 
$k\leq t-1$, then it does so; in this case the path $P_{t}$ continues
downwards until it reaches vertex $p_{i_{t}}$, where it ends (cf.~Figure~\ref%
{planar-lower-bound-fig} for $t=3$ and $e_{i_{3}}=e_{3}$). Suppose now that $%
P_{t}$ can not reach the column of $P_{t}$ without passing through a vertex
of a path $P_{k}$, $k\leq t-1$ (cf.~Figure~\ref{planar-lower-bound-fig} for $%
t=5$ and $e_{i_{5}}=e_{1}$). Then, $P_{t}$ continues on the row of vertex $%
a_{t}$ to the right as much as possible (say, until vertex $b_{t}$), such
that it does not reach any vertex of a path $P_{k}$, $k\leq t-1$. In this
case the path $P_{t}$ continues from vertex $b_{t}$ downwards as much as
possible until it reaches a vertex $c_{t}$ that is not neighbored to its
right to any vertex of a path $P_{k}$, $k\leq t-1$ (cf.~Figure~\ref%
{planar-lower-bound-fig} for $t=5$ and $e_{i_{5}}=e_{1}$). Furthermore $%
P_{t} $ continues from vertex $c_{t}$ to the right as much as possible until
it reaches a vertex $d_{t}$ that is not neighbored from above to any vertex
of a path $P_{k}$, $k\leq t-1$. Then, $P_{t}$ continues from $d_{t}$ in a
similar way until it reaches the column of vertex $p_{i_{t+1}}$ (cf.~Figure~\ref{planar-lower-bound-fig} for $t=5$, $e_{i_{5}}=e_{1}$, and $%
e_{i_{6}}=e_{6}$), and then it continues downwards until it reaches $%
p_{i_{t+1}}$, where $P_{t}$ ends. Note that, by definition of the edge set $%
S $, there exist at least $2n$ columns of the grid between any two edges of
the set $S$. Furthermore there exist $n-1$ rows of the grid below every edge
of $S$. Thus, since there exist at most $t-1\leq n-2$ previous paths $P_{k}$, 
$k\leq t-1$, it follows that there exists always enough space for the path 
$P_{t}$ in the grid to (a) reach vertex $d_{t}$ and (b) continue from $d_{t}$
until it reaches vertex $p_{i_{t+1}}$, where $P_{t}$ ends.

Assume now that $i_{t+1}<i_{t}$, i.e.~vertex $p_{i_{t+1}}$ lies to the left
of vertex $q_{i_{t}}$ on the $n$th row of the grid. In this case, when we
start the path $P_{t}$ at vertex $q_{i_{t}}$, we first move one edge
downwards and then two edges to the left (cf.~Figure~\ref%
{planar-lower-bound-fig} for $t=2$ and $e_{i_{2}}=e_{5}$, as well as for $%
t=4 $ and $e_{i_{4}}=e_{4}$). After that point we continue constructing the
path $P_{t}$ similarly to the case where $i_{t+1}>i_{t}$ (cf.~Figure~\ref%
{planar-lower-bound-fig}).

Therefore, we can construct in this way all the paths $P_{1},P_{2},\ldots
,P_{n-1}$, such that no two of these paths share a common vertex, and thus
the path $P=(p_{i_{1}},q_{i_{1}},P_{1},p_{i_{2}},$ $q_{i_{2}},P_{2},\ldots
p_{i_{n-1}},q_{i_{n-1}},P_{n-1},p_{i_{n}},q_{i_{n}})$ is a simple path of $G$
that visits all the edges of $S$ in the order of the permutation $\pi $. An
example of the construction of such a path $P$ is given in Figure~\ref%
{planar-lower-bound-fig}. In this example $S=(e_{1},e_{2},\ldots ,e_{6})$
and $\pi =(e_{2},e_{5},e_{3},e_{4},e_{1},e_{6})$. That is, using the above
notation, $i_{1}=2$, $i_{2}=5$, $i_{3}=3$, $i_{4}=4$, $i_{5}=1$, and $%
i_{6}=6 $. In this figure we also depict for $t=5$ the vertices $%
a_{t},b_{t},c_{t}$ that we defined in the above construction of the path $%
P_{i_{t}}$.

\begin{figure}[tbh]
\centering{\includegraphics[scale=0.65]{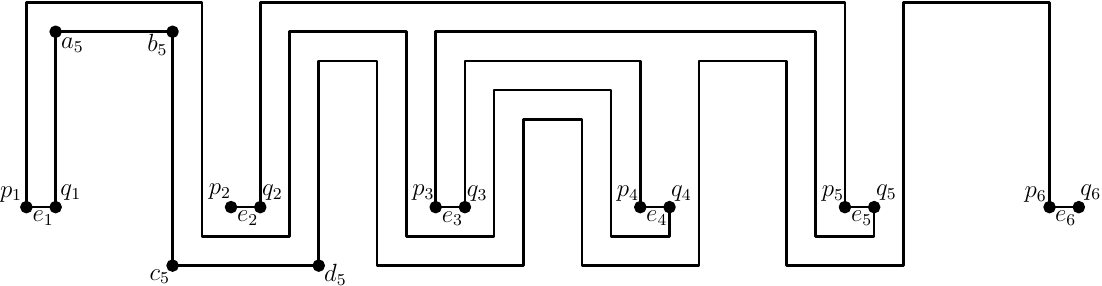}}
\caption{The edge-kernel $S=(e_{1},e_{2},\ldots ,e_{n}\}$ of the grid graph
with dimension $2n^{2} \times 2n$, where $n=6$, and a path $P$ that visits
the edges of $S$ in the order of the permutation $\protect\pi =
(e_{i_{1}},e_{i_{2}},e_{i_{3}},e_{i_{4}},e_{i_{5}},e_{i_{6}}) =
(e_{2},e_{5},e_{3},e_{4},e_{1},e_{6})$.}
\label{planar-lower-bound-fig}
\end{figure}

Since such a path $P$ exists for every permutation $\pi $ of the edges of
the set $S$, it follows by Definition~\ref{def:kernel} that $S$ is an edge-kernel
of $G$, where $G$ is a planar graph. Finally, since $G=(V,E)$ has by
construction $|V|=4n^{3}$ vertices and $|S|=n$, it follows that the size of
the edge-kernel $S$ is $\Omega (|V|^{\frac{1}{3}})$. This completes the
proof of the lemma.
\qquad
\end{proof}

\subsection{Computing the Cost}
\label{cost-computation-subsec}

\subsubsection{Hardness of Approximation}

Consider a boolean formula $\phi $ in conjunctive normal form with two
literals in every clause ($2$-CNF). Let $\tau $ be a truth assignment of the
variables of $\phi $ and $\alpha =(\ell _{1}\vee \ell _{2})$ be a clause of $%
\phi $. Then $\alpha $ is \emph{XOR-satisfied} (or \emph{NAE-satisfied}) in $%
\tau $, if one of the literals $\{\ell _{1},\ell _{2}\}$ of the clause $%
\alpha $ is true in $\tau $ and the other one is false in $\tau $. The
number of clauses of $\phi $ that are XOR-satisfied in $\tau $ is denoted by 
$|\tau(\phi )|$. The formula $\phi $ is \emph{XOR-satisfiable} (or \emph{NAE-satisfiable}) if there exists a truth assignment $\tau $ of $\phi $ such
that every clause of $\phi $ is XOR-satisfied in $\tau $. The \emph{Max-XOR}
problem (also known as the \emph{Max-NAE-2-SAT} problem) is the following
maximization problem: given a $2$-CNF formula $\phi $, compute the greatest
number of clauses of $\phi $ that can be simultaneously XOR-satisfied in a
truth assignment $\tau $, i.e.~compute the greatest value for $|\tau(\phi )|$.
The \emph{Max-XOR(}$k$\emph{)} problem is the special case of the 
Max-XOR problem, where every variable of the input formula $\phi $ appears in
at most~$k$ clauses of $\phi $. 
It is known that a special case of Max-XOR($3$), namely the \emph{monotone Max-XOR($3$)} problem, is APX-hard 
(i.e.~it does not admit a PTAS unless P=NP~\cite{KMSV99,CKS01}), as the next lemma states~\cite{AlimontiKann97}.
In this special case of the problem, the input formula $\phi$ is monotone, i.e.~every variable appears not negated in the formula. 
The monotone Max-XOR($3$) problem essentially encodes the \emph{Max-Cut} problem on 3-regular (i.e.~cubic) graphs, which is known to be APX-hard~\cite{AlimontiKann97}.
\begin{lemma}[\hspace{-0,001cm}\cite{AlimontiKann97}]
\label{Max_XOR-3-hard-lem}The (monotone) Max-XOR($3$) problem is APX-hard.
\end{lemma}

Now we provide a reduction from the Max-XOR$(3)$ problem to the problem of
computing $\kappa(G,\reach,d(G))$. Let $\phi $ be an instance formula of Max-XOR$(3)$ with $n$
variables $x_{1},x_{2},\ldots ,x_{n}$ and $m$ clauses. Since every variable $%
x_{i}$ appears in $\phi $ (either as $x_{i}$ or as $\overline{x_{i}}$) in at
most $3$ clauses, it follows that $m\leq \frac{3}{2}n$. We will construct
from $\phi $ a graph $G_{\phi }$ having length of a directed cycle at most $%
2 $. Then, as we prove in Theorem~\ref{cost-diameter-upper-lower-bound-thm}, 
$\kappa(G_{\phi },\reach,d(G_{\phi }))\leq 39n-4m-2k$ if and only if there
exists a truth assignment $\tau $ of $\phi $ with $|\tau(\phi )|\geq k$, i.e.~$%
\tau $ XOR-satisfies at least $k$ clauses of $\phi $. Since $\phi $ is an
instance of Max-XOR$(3)$, we can replace every clause $(\overline{x_{i}}\vee 
\overline{x_{j}})$ by the clause $(x_{i}\vee x_{j})$ in $\phi $, since $(%
\overline{x_{i}}\vee \overline{x_{j}})=(x_{i}\vee x_{j})$ in XOR.
Furthermore, whenever $(\overline{x_{i}}\vee x_{j})$ is a clause of $\phi $,
where $i<j$, we can replace this clause by $(x_{i}\vee \overline{x_{j}})$,
since $(\overline{x_{i}}\vee x_{j})=(x_{i}\vee \overline{x_{j}})$ in XOR.
Thus, we can assume without loss of generality that every clause of $\phi $
is either of the form $(x_{i}\vee x_{j})$ or $(x_{i}\vee \overline{x_{j}})$,
where $i<j$.

For every $i=1,2,\ldots ,n$ we construct the graph $G_{\phi ,i}$ of Figure~\ref{variable-gadget-fig}. Note that the diameter of $G_{\phi ,i}$ is $%
d(G_{\phi ,i})=9$ and the maximum length of a directed cycle in $G_{\phi ,i}$
is $2$. In this figure, we call the induced subgraph of $G_{\phi ,i}$ on the 
$13$ vertices $\{s^{x_{i}},u_{1}^{x_{i}},\ldots
,u_{6}^{x_{i}},v_{1}^{x_{i}},\ldots ,v_{6}^{x_{i}}\}$ the \emph{trunk} of $%
G_{\phi ,i}$. Furthermore, for every $p\in \{1,2,3\}$, we call the induced
subgraph of $G_{\phi ,i}$ on the $5$ vertices $%
\{u_{7,p}^{x_{i}},u_{8,p}^{x_{i}},v_{7,p}^{x_{i}},v_{8,p}^{x_{i}},t_{p}^{x_{i}},\} 
$ the $p$\emph{th branch} of $G_{\phi ,i}$. Finally, we call the edges $%
u_{6}^{x_{i}}u_{7,p}^{x_{i}}$ and $v_{6}^{x_{i}}v_{7,p}^{x_{i}}$ the \emph{transition edges} of the $p$th branch of $G_{\phi ,i}$. Furthermore, for
every $i=1,2,\ldots ,n$, let $r_{i}\leq 3$ be the number of clauses in which
variable $x_{i}$ appears in $\phi $. For every $1\leq p\leq r_{i}$, we
assign the $p$th appearance of the variable $x_{i}$ (either as $x_{i}$ or as 
$\overline{x_{i}}$) in a clause of $\phi $ to the $p$th branch of~$G_{\phi
,i}$.

Consider now a clause $\alpha =(\ell _{i}\vee \ell _{j})$ of $\phi $, where $%
i<j$. Then, by our assumptions on $\phi $, it follows that $\ell _{i}=x_{i}$
and $\ell _{j}\in \{x_{j},\overline{x_{j}}\}$. Assume that the literal $\ell
_{i}$ (resp.~$\ell _{j}$) of the clause $\alpha $ corresponds to the $p$th
(resp.~to the $q$th) appearance of the variable $x_{i}$ (resp.~$x_{j}$) in $%
\phi $. Then we identify the vertices of the $p$th branch of $G_{\phi ,i}$
with the vertices of the $q$th branch of $G_{\phi ,j}$ as follows. If $\ell
_{j}=x_{j}$ then we identify the vertices $%
u_{7,p}^{x_{i}},u_{8,p}^{x_{i}},v_{7,p}^{x_{i}},v_{8,p}^{x_{i}},t_{p}^{x_{i}} 
$ with the vertices $%
v_{7,q}^{x_{j}},v_{8,q}^{x_{j}},u_{7,q}^{x_{j}},u_{8,q}^{x_{j}},t_{q}^{x_{j}} 
$, respectively (cf.~Figure~\ref{clause-gadget-fig-1}). Otherwise, if $\ell
_{j}=\overline{x_{j}}$ then we identify the vertices $%
u_{7,p}^{x_{i}},u_{8,p}^{x_{i}},v_{7,p}^{x_{i}},v_{8,p}^{x_{i}},t_{p}^{x_{i}} 
$ with the vertices $%
u_{7,q}^{x_{j}},u_{8,q}^{x_{j}},v_{7,q}^{x_{j}},v_{8,q}^{x_{j}},t_{q}^{x_{j}} 
$, respectively (cf.~Figure~\ref{clause-gadget-fig-2}). This completes the
construction of the graph $G_{\phi }$. Note that, similarly to the graphs $%
G_{\phi ,i}$, $1\leq i\leq n$, the diameter of $G_{\phi }$ is $d(G_{\phi
})=9 $ and the maximum length of a directed cycle in $G_{\phi }$ is $2$.
Furthermore, note that for each of the $m$ clauses of $\phi $, one branch of
a gadget $G_{\phi ,i}$ coincides with one branch of a gadget $G_{\phi ,j}$,
where $1\leq i<j\leq n$, while every $G_{\phi ,i}$ has three branches.
Therefore $G_{\phi }$ has exactly $3n-2m$ branches which belong to only one
gadget $G_{\phi ,i}$, and $m$ branches that belong to two gadgets $G_{\phi
,i},G_{\phi ,j}$.

\begin{figure}[tbh]
\centering\includegraphics[scale=0.68]{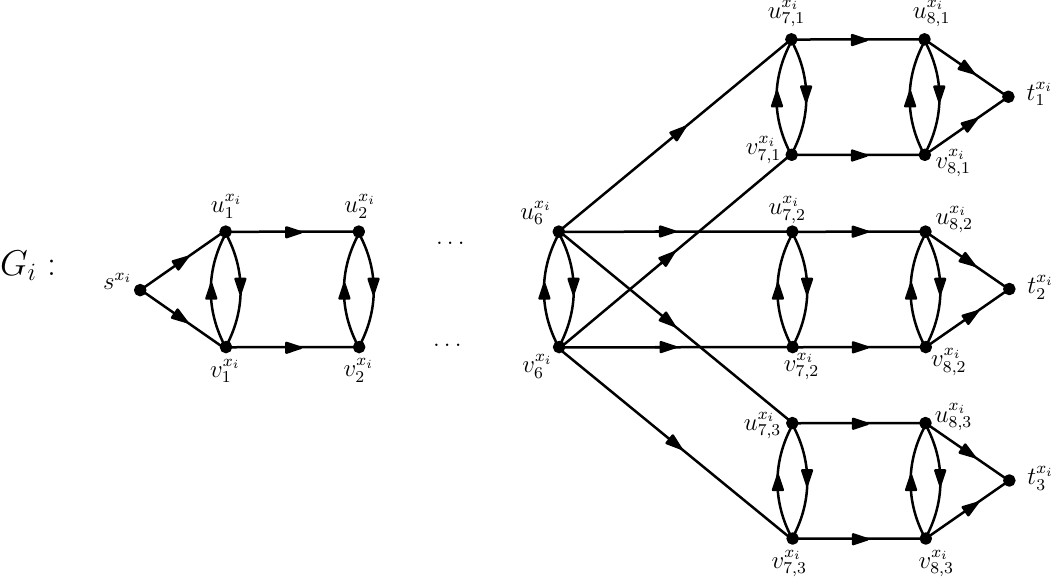}
\caption{The gadget $G_{\protect\phi ,i}$ for the variable $x_{i}$.}
\label{variable-gadget-fig}
\end{figure}

\begin{theorem}
\label{cost-diameter-upper-lower-bound-thm}There exists a truth assignment $%
\tau $ of $\phi $ with $|\tau(\phi )|\geq k$ if and only if $\kappa(G_{\phi
},reach,d(G_{\phi }))\leq 39n-4m-2k$.
\end{theorem}

\begin{proof}
($\Rightarrow $) Assume that there is a truth assignment $\tau $ that
XOR-satisfies $k$ clauses of $\phi $. We construct a labeling $\lambda $ of $%
G_{\phi }$ with cost $39n-4m-2k$ as follows. Let $i=1,2,\ldots ,n$. If $%
x_{i}=0$ in $\tau $, we assign labels to the edges of the trunk of $G_{\phi
,i}$ as in Figure~\ref{labeling-x-0-fig}. Otherwise, if $x_{i}=1$ in $\tau$, we assign labels to the edges of the trunk of $G_{\phi ,i}$ as in Figure~\ref{labeling-x-1-fig}. We now continue the labeling $\lambda $ as follows.
Consider an arbitrary clause $\alpha =(\ell _{i}\vee \ell _{j})$ of $\phi $,
where $i<j$. Recall that $\ell _{i}=x_{i}$ and $\ell _{j}\in \{x_{j},%
\overline{x_{j}}\}$. Assume that the literal $\ell _{i}$ (resp.~$\ell _{j}$)
of the clause $\alpha $ corresponds to the $p$th (resp.~to the $q$th)
appearance of variable $x_{i}$ (resp.~$x_{j}$) in $\phi $. Then, by the
construction of $G_{\phi }$, the $p$th branch of $G_{\phi ,i}$ coincides
with the $q$th branch of $G_{\phi ,j}$.

\begin{figure}[tbh]
\centering%
\subfigure[]{ \label{labeling-x-0-fig}
\includegraphics[scale=0.68]{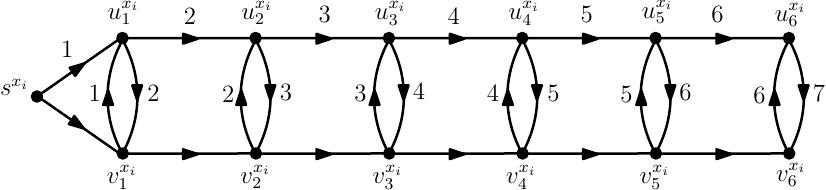}} 
\subfigure[]{ \label{labeling-x-1-fig}
\includegraphics[scale=0.68]{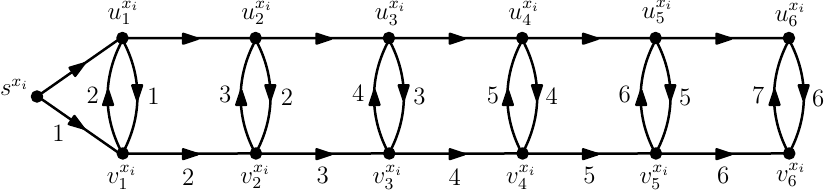}}
\caption{The labels of the edges of the trunk of $G_{\protect\phi ,i}$,
where (a)~$x=0$ and (b)~$x=1$.}
\label{labeling-x-0-1-fig}
\end{figure}

Assume that $\ell _{j}=x_{j}$ (cf.~Figure~\ref{clause-gadget-fig-1}). Then
by our construction $u_{7,p}^{x_{i}}=v_{7,q}^{x_{j}}$, $%
u_{8,p}^{x_{i}}=v_{8,q}^{x_{j}}$, $v_{7,p}^{x_{i}}=u_{7,q}^{x_{j}}$, $%
v_{8,p}^{x_{i}}=u_{8,q}^{x_{j}}$, and $t_{p}^{x_{i}}=t_{q}^{x_{j}}$ (cf.
Figure~\ref{clause-gadget-fig-1}). Let $\alpha $ be XOR-satisfied in $\tau $, i.e.~$x_{i}=\overline{x_{j}}$. If $x_{i}=\overline{x_{j}}=0$ then we label
the edges of the $p$th branch of $G_{\phi ,i}$ (equivalently, the edges of
the $q$th branch of $G_{\phi ,j}$), the transition edges of the $p$th branch
of $G_{\phi ,i}$, and the transition edges of the $q$th branch of $G_{\phi
,j}$, as illustrated in Figure~\ref{assignment-fig-1}. In the symmetric case
where $x_{i}=\overline{x_{j}}=1$ we label these edges in the same way as in
Figure~\ref{assignment-fig-1}, with the only difference that we exchange the
role of $u$'s and $v$'s. Let now $\alpha $ be XOR-unsatisfied in $\tau $,
i.e.~$x_{i}=x_{j}$. If $x_{i}=x_{j}=0$ then we label the edges of the $p$th
branch of $G_{\phi ,i}$ (equivalently, the edges of the $q$th branch of $%
G_{\phi ,j}$), the transition edges of the $p$th branch of $G_{\phi ,i}$,
and the transition edges of the $q$th branch of $G_{\phi ,j}$, as
illustrated in Figure~\ref{assignment-fig-2}. In the symmetric case where $%
x_{i}=x_{j}=1$ we label these edges in the same way as in Figure~\ref%
{assignment-fig-2}, with the only difference that we exchange the role of $u$'s and $v$'s. For the case where $\ell _{j}=\overline{x_{j}}$ we label the
edges of Figure~\ref{clause-gadget-fig-2} similarly to the case where $\ell
_{j}=x_{j}$ (cf.~Figure~\ref{assignment-fig}). 
\begin{figure}[tbh]
\centering%
\subfigure[]{ \label{clause-gadget-fig-1}
\includegraphics[scale=0.68]{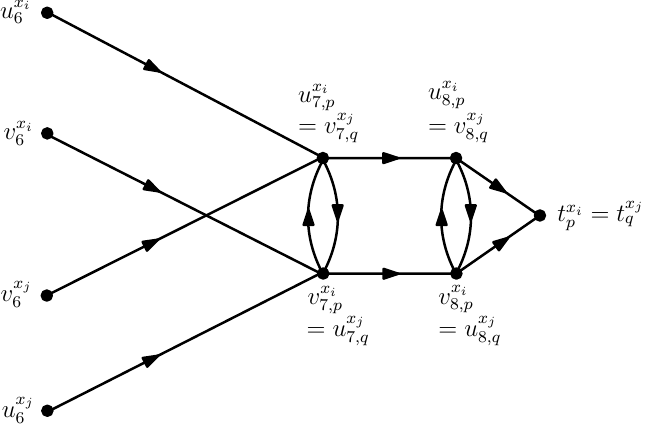}} \hspace{0.2cm} 
\subfigure[]{ \label{clause-gadget-fig-2}
\includegraphics[scale=0.68]{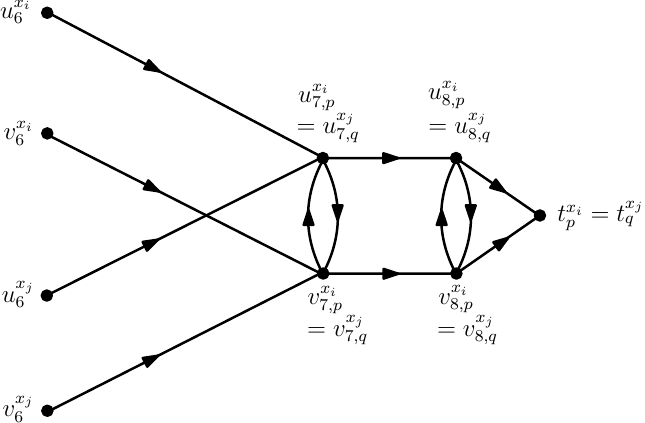}}
\caption{The gadgets for (a)~the clause $(x_{i}\vee x_{j})$ and (b)~the
clause $(x_{i}\vee \overline{x_{j}})$, where $x_{i}$ appears in the $p$th
branch of $G_{\protect\phi ,i}$ and $x_{j}$ (resp.~$\overline{x_{j}}$)
appears in the $q$th branch of $G_{\protect\phi ,j}$.}
\label{clause-gadget-fig}
\end{figure}

Finally consider any of the $3n-2m$ branches that belong to only one gadget $%
G_{\phi ,i}$, where $1\leq i\leq n$. Let this be the $p$th branch of $%
G_{\phi ,i}$. If $x_{i}=0$ then we label the edges of this branch and its
transition edges as illustrated in Figure~\ref{assignment-fig-1} (by
ignoring in this figure the vertices $u_{6}^{x_{j}},v_{6}^{x_{j}}$). In the
symmetric case where $x_{i}=1$, we label these edges in the same way, with
the only difference that we exchange the role of $u$'s and $v$'s. This
finalizes the labeling $\lambda $ of $G_{\phi }$. It is easy to check that $%
\lambda $ preserves all reachabilities of $G_{\phi }$ and its greatest label
is $d$. 
\begin{figure}[tbh]
\centering%
\subfigure[]{ \label{assignment-fig-1}
\includegraphics[scale=0.68]{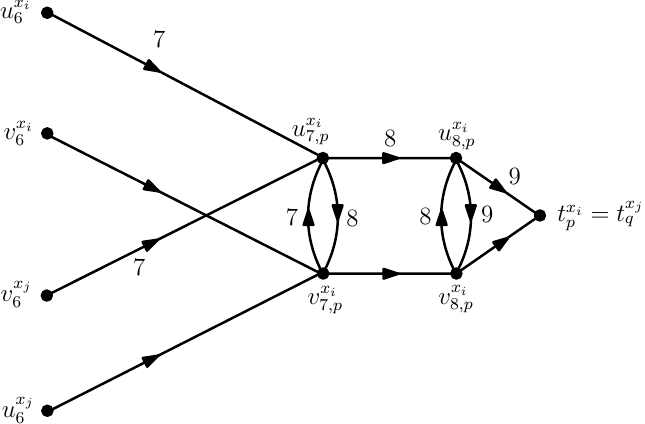}} 
\subfigure[]{ \label{assignment-fig-2}
\includegraphics[scale=0.68]{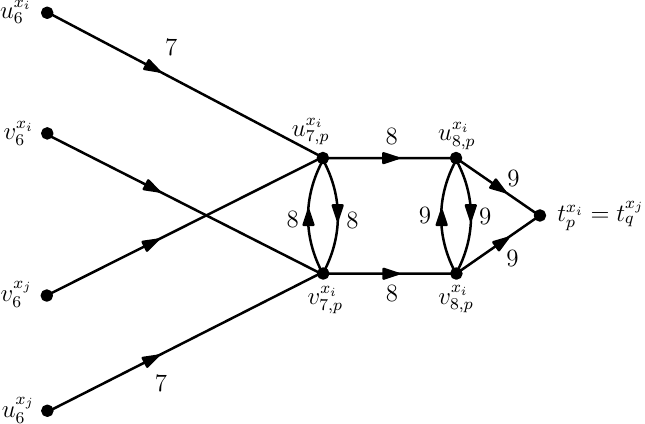}}
\caption{The labeling of the edges of Figure~\protect\ref%
{clause-gadget-fig-1} for the clause $\protect\alpha =(x_{i}\vee x_{j})$,
where (a)~$\protect\alpha $ is XOR-satisfied and $x_{i}=\overline{x_{j}}=0$
in $\protect\tau $ and (b)~$\protect\alpha $ is XOR-unsatisfied and $%
x_{i}=x_{j}=0$ in $\protect\tau $.}
\label{assignment-fig}
\end{figure}

Summarizing, for every $i\in \{1,2,\ldots ,n\}$, the edges of the trunk of $%
G_{\phi ,i}$ are labeled with $18$ labels (cf.~Figure~\ref%
{labeling-x-0-1-fig}), and thus $\lambda $ uses in total $18n$ labels for
the trunks of all $G_{\phi ,i}$, $i\in \{1,2,\ldots ,n\}$. Furthermore, for
every $i\in \{1,2,\ldots ,n\}$ and every $p\in \{1,2,3\}$, $\lambda $ uses $%
1 $ label for the two transition edges of the $p$th branch of $G_{\phi ,i}$
(cf.~Figure~\ref{assignment-fig}), and thus $\lambda $ uses in total $3n$
labels for the transition edges of all $G_{\phi ,i}$, $i\in \{1,2,\ldots
,n\} $. Moreover, for each of the $3n-2m$ branches that belong to only one
gadget $G_{\phi ,i}$, where $1\leq i\leq n$, $\lambda $ uses $6$ labels for
the edges of this branch of $G_{\phi ,i}$, and thus $\lambda $ uses in total 
$6(3n-2m)$ labels for all these $3n-2m$ branches. Finally consider any of
the remaining $m$ branches of $G_{\phi }$, each of which corresponds to a
clause $\alpha $ of $\phi $ (i.e.~this branch belongs simultaneously to a
gadget $G_{\phi ,i}$ and a gadget $G_{\phi ,j}$, where $1\leq i<j\leq n$).
If $\alpha $ is XOR-satisfied in $\tau $, then $\lambda $ uses $6$ labels
for the edges of this branch (cf.~for example Figure~\ref{assignment-fig-1}). Otherwise, if $\alpha $ is XOR-unsatisfied in $\tau $, then $\lambda $
uses $8$ labels for the edges of this branch (cf.~for example Figure~\ref%
{assignment-fig-2}). Therefore, since $\tau $ XOR-satisfies by assumption $k$
of the $m$ clauses of $\phi $, it follows $\lambda $ uses in total $%
18n+3n+6(3n-2m)+6k+8(m-k)=39n-4m-2k$ labels, and thus $\kappa(G_{\phi
},\reach,d(G_{\phi }))\leq 39n-4m-2k$.

\medskip

($\Leftarrow $) Assume that $\kappa(G_{\phi },\reach,d(G_{\phi }))\leq
39n-4m-2k$ and let $\lambda $ be a labeling of $G_{\phi }$ 
that maintains all reachabilities and has minimum cost (i.e.~has the smallest number of labels); that is, 
$|\lambda|\leq 39n-4m-2k$. 
Let $i\in \{1,2,\ldots ,n\}$. Note that for every $z\in
\{1,2,\ldots ,6\}$, the vertices $u_{z}^{x_{i}}$ and $v_{z}^{x_{i}}$ reach
each other in $G_{\phi }$ with a unique path (of length one). Therefore,
each of the directed edges $\left\langle
u_{z}^{x_{i}}v_{z}^{x_{i}}\right\rangle $ and $\left\langle
v_{z}^{x_{i}}u_{z}^{x_{i}}\right\rangle $, where $z\in \{1,2,\ldots ,6\}$,
receives at least one label in every labeling, and thus also in $\lambda $.
Similarly it follows that each of the directed edges $\left\langle
u_{z,p}^{x_{i}}v_{z,p}^{x_{i}}\right\rangle $ and $\left\langle
u_{z,p}^{x_{i}}v_{z,p}^{x_{i}}\right\rangle $, where $z\in \{7,8\}$ and $%
p\in \{1,2,3\}$, receives at least one label in every labeling, and thus
also in $\lambda $.

For every $i\in \{1,2,\ldots ,n\}$, define now the two paths $%
P_{i}=(s^{x_{i}},u_{1}^{x_{i}},u_{2}^{x_{i}},\ldots ,u_{6}^{x_{i}})$ and $%
Q_{i}=(s^{x_{i}},v_{1}^{x_{i}},v_{2}^{x_{i}},\ldots ,v_{6}^{x_{i}})$.
Furthermore, for every $p\in \{1,2,3\}$, define the paths $%
P(i,p)=(P_{i},u_{7,p}^{x_{i}},u_{8,p}^{x_{i}},t_{p}^{x_{i}})$ and $%
Q(i,p)=(Q_{i},v_{7,p}^{x_{i}},v_{8,p}^{x_{i}},t_{p}^{x_{i}})$. Note that $%
P(i,p)$ and $Q(i,p)$ are the only two paths in $G_{\phi }$ from $s^{x_{i}}$
to $t_{p}^{x_{i}}$ with distance $d(G_{\phi })=9$. Thus, since $\lambda $
preserves all reachabilities of $G_{\phi }$ with maximum label $9$, it
follows that for every $i\in \{1,2,\ldots ,n\}$ and every $p\in \{1,2,3\}$,
the edges of $P(i,p)$ or the edges of $Q(i,p)$ are labeled with the labels $%
1,2,\ldots ,9$ in $\lambda $.

Assume that there exists an ${i\in \{1,2,\ldots ,n\}}$ such that all edges
of the path $P_{i}$ and all edges of the path $Q_{i}$ are labeled in $\lambda$. 
Note that, if there exists no value $p\in\{1,2,3\}$ such that all edges of $P(i,p)$ (resp.~of $Q(i,p)$) are labeled, 
then we can remove all labels from $P(i,p)$ (resp.~from $Q(i,p)$) and construct another labeling $\lambda'$ that still maintains all reachabilities 
of $G_{\phi}$ but has fewer labels than $\lambda$, which is a contradiction to the minimality assumption of $\lambda$. 
Therefore, there must exist values $p,q\in\{1,2,3\}$ such that all edges of $P(i,p)$ and all edges of $Q(i,q)$ are labeled in~$\lambda$.
Then, in both cases where ${p=q}$ and ${p\neq q}$, we modify $%
\lambda $ into a labeling $\lambda ^{\prime }$ as follows. We remove the
labels from the seven edges of the path $(Q_{i},v_{7,q}^{x_{i}})$, and we
add labels (if they do not already have labels) to the six edges $%
\langle u_{6}^{x_{i}}u_{7,z}^{x_{i}}\rangle ,\langle
u_{7,z}^{x_{i}}u_{8,z}^{x_{i}}\rangle ,\langle
u_{8,z}^{x_{i}}t_{z}^{x_{i}}\rangle $, where $z\in \{1,2,3\}\setminus \{p\}$. Note that, in this new labeling $\lambda ^{\prime }$, we can always
preserve all reachabilities of the vertices by choosing the appropriate
labels for the edges $\left\langle u_{1}^{x_{i}}v_{1}^{x_{i}}\right\rangle
,\left\langle v_{1}^{x_{i}}u_{1}^{x_{i}}\right\rangle ,\left\langle
u_{2}^{x_{i}}v_{2}^{x_{i}}\right\rangle ,\left\langle
v_{2}^{x_{i}}u_{2}^{x_{i}}\right\rangle ,\ldots ,\left\langle
u_{6}^{x_{i}}v_{6}^{x_{i}}\right\rangle ,$ $\langle v_{6}^{x_{i}}u_{6}^{x_{i}}\rangle ,\langle
u_{7,z}^{x_{i}}v_{7,z}^{x_{i}}\rangle ,\langle
v_{7,z}^{x_{i}}u_{7,z}^{x_{i}}\rangle ,\langle
u_{8,z}^{x_{i}}v_{8,z}^{x_{i}}\rangle ,\langle
v_{8,z}^{x_{i}}u_{8,z}^{x_{i}}\rangle $, where $z\in \{1,2,3\}$, cf. for
example the labelings of Figures~\ref{labeling-x-0-1-fig} and~\ref%
{assignment-fig}. However, by construction, the new labeling $\lambda
^{\prime }$ uses a smaller number of labels than the initial labeling $%
\lambda $, which is a contradiction. Therefore,
we may assume without loss of generality that for every $i\in \{1,2,\ldots
,n\}$, it is not the case that all edges of both paths $P_{i}$ and $Q_{i}$
are labeled in $\lambda $, i.e.~either all edges of $P_{i}$ or all edges of $%
Q_{i}$ are labeled in $\lambda $.

We now construct a truth assignment $\tau $ for the formula $\phi $ as
follows. For every $i\in \{1,2,\ldots ,n\}$, if all edges of the path $P_{i}$
are labeled in $\lambda $, then we define $x_{i}=0$ in $\tau $. Otherwise,
if all edges of the path $Q_{i}$ are labeled in $\lambda $, then we define $%
x_{i}=1$ in $\tau $. We will prove that $|\tau(\phi )|\geq k$, i.e.~that $\tau $
XOR-satisfies at least $k$ clauses of the formula $\phi $.

Let $i\in \{1,2,\ldots ,n\}$. Recall that each of the directed edges $%
\left\langle u_{z}^{x_{i}}v_{z}^{x_{i}}\right\rangle $ and $\left\langle
v_{z}^{x_{i}}u_{z}^{x_{i}}\right\rangle $, where $z\in \{1,2,\ldots ,6\}$,
receives at least one label in $\lambda $. Therefore, since all six edges of 
$P_{i}$ or all six edges of $Q_{i}$ are labeled in $\lambda $, it follows
that $\lambda $ uses for the trunk of $G_{\phi ,i}$ at least $18$ labels.
Thus, $\lambda $ uses in total at least $18n$ labels for the trunks of all $%
G_{\phi ,i}$, $i\in \{1,2,\ldots ,n\}$.

Let now $p\in \{1,2,3\}$. Then, since $%
P(i,p)=(P_{i},u_{7,p}^{x_{i}},u_{8,p}^{x_{i}},t_{p}^{x_{i}})$ and $%
Q(i,p)=(Q_{i},v_{7,p}^{x_{i}},v_{8,p}^{x_{i}},t_{p}^{x_{i}})$ are the only
two paths in $G_{\phi }$ from $s^{x_{i}}$ to $t_{p}^{x_{i}}$ with distance $%
d(G_{\phi })=9$, it follows that $\lambda $ uses at least one label for the
pair of the transition edges $\{\langle u_{6}^{x_{i}}u_{7,p}^{x_{i}}\rangle
,$ $\langle v_{6}^{x_{i}}v_{7,p}^{x_{i}}\rangle \}$ of the $p$th branch of $%
G_{\phi ,i}$. Thus, $\lambda $ uses in total at least $3n$ labels for the
transition edges of all $G_{\phi ,i}$, $i\in \{1,2,\ldots ,n\}$.

Consider an arbitrary branch of $G_{\phi }$, e.g.~the $p$th branch of $%
G_{\phi ,i}$, where $i\in \{1,2,\ldots ,n\}$ and $p\in \{1,2,3\}$. Since $%
P(i,p)$ and $Q(i,p)$ are the only two paths in $G_{\phi }$ from $s^{x_{i}}$
to $t_{p}^{x_{i}}$ with distance $d(G_{\phi })=9$, it follows that $\lambda $
assigns at least one label to each of the edges $\{\langle
u_{7,p}^{x_{i}}u_{8,p}^{x_{i}}\rangle ,\langle
u_{8,p}^{x_{i}}t_{p}^{x_{i}}\rangle \}$, or at least one label to each of
the edges $\{\langle v_{7,p}^{x_{i}}v_{8,p}^{x_{i}}\rangle ,\langle
v_{8,p}^{x_{i}}t_{p}^{x_{i}}\rangle \}$. Furthermore recall that each of the
edges $\left\langle u_{z,p}^{x_{i}}v_{z,p}^{x_{i}}\right\rangle $ and $%
\left\langle v_{z,p}^{x_{i}}u_{z,p}^{x_{i}}\right\rangle $, where $z\in
\{7,8\}$, receives at least one label in $\lambda $. Therefore, $\lambda $
uses at least $6$ labels for an arbitrary branch of $G_{\phi }$.

Consider now one of the clauses $\alpha =(\ell _{i}\vee \ell _{j})$ of $\phi 
$ that are not XOR-satisfied in $\tau $ that we defined above. Note that
there exist exactly $m-|\tau(\phi )|$ such clauses in $\phi $. Let $i<j$, and
thus $\ell _{i}=x_{i}$ and $\ell _{j}\in \{x_{j},\overline{x_{j}}\}$. Assume
that the literal $\ell _{i}$ (resp.~$\ell _{j}$) of the clause $\alpha $
corresponds to the $p$th (resp.~to the $q$th) appearance of variable $x_{i}$
(resp.~$x_{j}$) in $\phi $. Then, by the construction of $G_{\phi }$, the $p$th branch of $G_{\phi ,i}$ coincides with the $q$th branch of $G_{\phi ,j}$.
Suppose first that $\ell _{j}=x_{j}$. Then $x_{i}=x_{j}$, since $\alpha $ is
not XOR-satisfied in $\tau $. By the construction of the truth assignment $%
\tau $ from the labeling $\lambda $, it follows that either all edges of $%
P(i,p)$ and all edges of $P(j,q)$ are labeled in $\lambda $ (in the case
where $x_{i}=x_{j}=0$), or all edges of $Q(i,p)$ and all edges of $Q(j,q)$
are labeled in $\lambda $ (in the case where $x_{i}=x_{j}=1$). Since $\ell
_{j}=x_{j}$, note by the construction of $G_{\phi }$ that the last two edges
of $P(i,p)$ are different from the last two edges of $P(j,q)$, while the
last two edges of $Q(i,p)$ are different from the last two edges of $Q(j,q)$. Therefore, since each of the edges $\langle
u_{z,p}^{x_{i}}v_{z,p}^{x_{i}}\rangle $ and $\langle
v_{z,p}^{x_{i}}u_{z,p}^{x_{i}}\rangle $, where $z\in \{7,8\}$, receives at
least one label in $\lambda $, it follows that $\lambda $ uses for the $p$th
branch of $G_{\phi ,i}$ (equivalently for the $q$th branch of $G_{\phi ,j}$)
at least $8$ labels, if $\ell _{j}=x_{j}$.

Suppose now that $\ell _{j}=\overline{x_{j}}$. Then $x_{i}=\overline{x_{j}}$, since $\alpha $ is not XOR-satisfied in $\tau $. Similarly to the case
where $x_{i}=x_{j}$, it follows that either all edges of $P(i,p)$ and all
edges of $Q(j,q)$ are labeled in $\lambda $ (in the case where $x_{i}=%
\overline{x_{j}}=0$), or all edges of $Q(i,p)$ and all edges of $P(j,q)$ are
labeled in $\lambda $ (in the case where $x_{i}=\overline{x_{j}}=1$). Since $%
\ell _{j}=\overline{x_{j}}$, note by the construction of $G_{\phi }$ that
the last two edges of $P(i,p)$ are different from the last two edges of $%
Q(j,q)$, while the last two edges of $Q(i,p)$ are different from the last
two edges of $P(j,q)$. Therefore, since each of the edges $\langle
u_{z,p}^{x_{i}}v_{z,p}^{x_{i}}\rangle $ and $\langle
v_{z,p}^{x_{i}}u_{z,p}^{x_{i}}\rangle $, where $z\in \{7,8\}$, receives at
least one label in $\lambda $, it follows that $\lambda $ uses for the $p$th
branch of $G_{\phi ,i}$ (equivalently for the $q$th branch of $G_{\phi ,j}$)
at least $8$ labels, if $\ell _{j}=\overline{x_{j}}$.

Summarizing, $\lambda $ uses in total at least $18n$ labels for the edges of
the trunks of all $G_{\phi ,i}$, at least $3n$ labels for the transition
edges of all $G_{\phi ,i}$, at least $6$ labels for an arbitrary branch of $%
G_{\phi }$, and at least $8$ labels for each of the branches of $G_{\phi }$
that corresponds to a clause $\alpha $ of $\phi $ that is not XOR-satisfied
in $\tau $. Therefore, since $G_{\phi }$ has in total $3n-m$ branches and $%
\phi $ has $m-|\tau(\phi )|$ XOR-unsatisfied clauses in $\tau $, it follows
that $\lambda $ uses at least $18n+3n+6(3n-m-(m-|\tau(\phi )|))+8(m-|\tau(\phi
)|)=39n-4m-2|\tau(\phi )|$ labels. However $|\lambda |\leq 39n-4m-2k$ by
assumption. Therefore $39n-4m-2|\tau(\phi )|\leq |\lambda |\leq 39n-4m-2k$, and
thus $\tau $ XOR-satisfies $|\tau(\phi )|\geq k$ clauses in $\phi $. This
completes the proof of the theorem.
\qquad
\end{proof}

\medskip

Using Theorem~\ref{cost-diameter-upper-lower-bound-thm}, we are now ready to
prove the main theorem of this section.

\begin{theorem} [Hardness of Approximating the Temporal Cost]
\label{cost-diameter-Apx-hard-thm}The problem of computing $%
\kappa(G,reach,d(G))$ is APX-hard, even when the maximum length of a
directed cycle in $G$ is $2$.
\end{theorem}

\begin{proof}
Denote now by OPT$_{\text{Max-XOR}(3)}(\phi )$ the greatest number of
clauses that can be simultaneously XOR-satisfied by a truth assignment of $%
\phi $. Then Theorem~\ref{cost-diameter-upper-lower-bound-thm} implies that 
\begin{equation}
\kappa(G_{\phi },reach,d(G_{\phi }))\leq 39n-4m-2\cdot \text{OPT}_{\text{Max-XOR}(3)}(\phi )  \label{c-opt-upper-bound-eq}
\end{equation}%
Note that a random assignment XOR-satisfies each clause of $\phi $ with
probability $\frac{1}{2}$, and thus we can easily compute (even
deterministically) an assignment $\tau $ that XOR-satisfies $\frac{m}{2}$
clauses of $\phi $. Therefore OPT$_{\text{Max-XOR}(3)}(\phi )\geq \frac{m}{2}
$, and thus, since every variable $x_{i}$ appears in at least one clause of $%
\phi $, it follows that%
\begin{equation}
\frac{n}{2}\leq ~m\leq 2\cdot\text{OPT}_{\text{Max-XOR}(3)}(\phi )  \label{m-upper-bound-eq}
\end{equation}
Assume that there is a PTAS for computing $\kappa(G,\reach,d(G))$. Then,
for every $\varepsilon >0$ we can compute in polynomial time a labeling $%
\lambda $ for the graph $G_{\phi }$, such that 
\begin{equation}
|\lambda |\leq (1+\varepsilon )\cdot \kappa(G_{\phi },reach,d(G_{\phi }))
\label{PTAS-bound-eq}
\end{equation}%
Given such a labeling $\lambda $ we can compute by the sufficiency part\ ($%
\Leftarrow $) of the proof of Theorem~\ref{cost-diameter-upper-lower-bound-thm} a truth assignment $\tau $ of $\phi $
such that $39n-4m-2|\tau(\phi )|\leq |\lambda |$, i.e.%
\begin{equation}
2|\tau(\phi )|\geq 39n-4m-|\lambda |  \label{tau-lower-bound-eq}
\end{equation}%
Therefore it follows by (\ref{c-opt-upper-bound-eq}), (\ref{m-upper-bound-eq}), (\ref{PTAS-bound-eq}), and (\ref{tau-lower-bound-eq}) that%
\begin{eqnarray*}
2|\tau(\phi )| &\geq &39n-4m-(1+\varepsilon )\cdot \kappa(G_{\phi},reach,d(G_{\phi })) \\
&\geq &39n-4m-(1+\varepsilon )\cdot \left( 39n-4m-2\cdot \text{OPT}_{\text{Max-XOR}(3)}(\phi )\right) \\
&=&\varepsilon \left( 4m-39n\right) +2(1+\varepsilon )\cdot \text{OPT}_{\text{Max-XOR}(3)}(\phi ) \\
&\geq &\varepsilon \left( 4m-78m\right) +2(1+\varepsilon )\cdot \text{OPT}_{\text{Max-XOR}(3)}(\phi ) \\
&\geq &-74\varepsilon m+2(1+\varepsilon )\cdot \text{OPT}_{\text{Max-XOR}(3)}(\phi ) \\
&\geq &-74\varepsilon \cdot 2\text{OPT}_{\text{Max-XOR}(3)}(\phi)+2(1+\varepsilon )\cdot \text{OPT}_{\text{Max-XOR}(3)}(\phi ) \\
&=&2(1-73\varepsilon )\cdot \text{OPT}_{\text{Max-XOR}(3)}(\phi )
\end{eqnarray*}%
and thus%
\begin{equation}
|\tau(\phi )|\geq (1-73\varepsilon )\cdot \text{OPT}_{\text{Max-XOR}(3)}(\phi )
\label{PTAS-contradiction-eq}
\end{equation}%

That is, assuming a PTAS for computing $\kappa(G,\reach,d(G))$, we obtain a
PTAS for the Max-XOR$(3)$ problem, which is a contradiction by Lemma~\ref%
{Max_XOR-3-hard-lem}. Therefore computing $\kappa(G,\reach,d(G))$ is
APX-hard. Finally, since the graph $G_{\phi }$ that we constructed from the
formula $\phi $ has maximum length of a directed cycle at most $2$, it
follows that computing $\kappa(G,\reach,d(G))$ is APX-hard even if the
given graph $G$ has maximum length of a directed cycle at most $2$.
\qquad
\end{proof}

\subsubsection{Approximating the Cost}

In this section, we provide an approximation algorithm for computing $\kappa(G,reach,d(G))$, which complements the hardness result of Theorem \ref{cost-diameter-Apx-hard-thm}. Given a digraph $G$ define, for every $u\in V$, $u$'s reachability number $r(u)= |\{v\in V: v\text{ is reachable from } u\}|$ and $r(G)=\sum_{u\in V} r(u)$, that is $r(G)$ is the total number of reachabilities in $G$.

\begin{theorem}
There is an $\frac{r(G)}{n-1}$-factor approximation algorithm for computing $\kappa(G,reach,d(G))$ on any weakly connected digraph $G$.
\end{theorem}
\begin{proof}
First of all, note that $\OPT\geq n-1$, where $\OPT$ is the cost of the optimal solution. The reason is that if a labeling labels less than $n-1$ edges then the subgraph of $G$ induced by the labeled edges is disconnected (not even weakly connected) thus clearly fails to preserve some reachabilities. To see this, take any two components $C_1$ and $C_2$. $G$ either has an edge from $C_1$ to $C_2$ or from $C_2$ to $C_1$ (or both). The two cases are symmetric so just consider the first one. Clearly some node from $C_1$ can reach some node from $C_2$ but this reachability has not been preserved by the labeling. 

Now consider the following labeling algorithm.
\begin{enumerate}
 \item For all $u\in V$, compute a BFS out-tree $T_u$ rooted at $u$.
 \item For all $T_u$, give to each edge at distance $i$ from the root label $i$.
 \item Output this labeling $\lambda$.
\end{enumerate}
Clearly, the maximum label used by $\lambda$ is $d(G)$: indeed if an edge $e$ was assigned some label $l>d(G)$ then this would imply that on some BFS out-tree $e$ appeared at distance $>d(G)$ which is a contradiction. Moreover, $\lambda$ preserves all reachabilities as for every $u$ the corresponding tree rooted at $u$ reaches all nodes that are reachable from $u$ and the described labeling clearly preserves the corresponding paths. Finally, we have that the cost paid by our algorithm is $\ALG=|\lambda|=r(G)$. To see this, notice that for all $u$ we use (i.e.~we label) precisely $r(u)$ edges in $T_u$, thus, in total, we use $\sum_{u\in V} r(u)=r(G)$ edges by definition of $r(G)$.

We conclude that
\begin{equation*}
\frac{\ALG}{\OPT} \leq \frac{r(G)}{n-1}\Rightarrow \ALG\leq \frac{r(G)}{n-1}\OPT
\end{equation*}
\qquad
\end{proof}

\section{Conclusions and Further Research}
\label{sec:conc}

There are many open problems related to the findings of the present work. We have considered several graph families in which the temporality of preserving all paths is very small (e.g.~2 for rings) and others in which it is very close to the worst possible (i.e.~$\Omega(n)$ for cliques and $\Omega(n^{1/3})$ for planar graphs). There are still many interesting graph families to be investigated like regular or bounded-degree graphs. Moreover, though it turned out to be a generic lower-bounding technique related to the existence of a large edge-kernel in the underlying graph $G$, we still do not know whether there are other structural properties of the underlying graph that could cause a growth of the temporality (i.e.~the absence of a large edge-kernel does not necessarily imply small temporality). Similar things hold also for the $\reach$ property. There are also many other natural connectivity properties subject to which optimization is to be performed that we haven't yet considered, like preserving a shortest path from $u$ to $v$ whenever $v$ is reachable from $u$ in $G$, or even depart from paths and require the preservation of more complex subgraphs (for some appropriate definition of ``preservation''). Another interesting direction which we didn't consider in this work is to set the optimization criterion to be the age of $\lambda$ e.g.~w.r.t. the $\ap$ or the $\reach$ connectivity properties. In this case, computing $\alpha(G,\text{all paths})$ is NP-hard, which can be proved by reduction from HAMPATH. On the positive side, it is easy to come up with a 2-factor approximation algorithm for $\alpha(G, reach, 2)$, where we have restricted the maximum number of labels of an edge (i.e.~the temporality) to be at most 2. 
Additionally, there seems to be great room for approximation algorithms (or even randomized algorithms) for all combinations of optimization parameters and connectivity constraints that we have defined so far, or even polynomial-time algorithms for specific graph families. Finally, it would be valuable to consider other models of temporal graphs and in particular models with succinct representations, that is models in which the labels of every edge are provided by some short function associated to that edge (in contrast to a complete description of all labels). Such examples are several probabilistic models and several periodic models which are worth considering.\\

\noindent \textbf{Acknowledgments.} We would like to thank the anonymous reviewers of this article and its preliminary versions. Their thorough reading and comments have helped us to improve our work substantially.

\end{document}